\documentclass[review]{elsarticle}
\usepackage{amsmath}
\usepackage[font={footnotesize,it}]{caption}
\usepackage[top=1.4in, bottom=1.4in, left=1in, right=1in]{geometry}

\DeclareCaptionStyle{italic}[justification=centering]{labelfont={bf},textfont={it},labelsep=colon}
\usepackage{hyperref}

\usepackage{graphicx,psfrag,epsf,textcomp}
\usepackage{enumerate, dsfont}
\usepackage{url,xcolor} 
\usepackage{booktabs, subfig, bm, paralist,mathpazo,tikz,todonotes,longtable,microtype}
\usepackage{color,soul,amsfonts,setspace,mathrsfs,mathtools,longtable,multirow}
\usepackage{natbib}
\usepackage{tikz}
\usepackage{amssymb}

\usetikzlibrary{shapes.geometric, arrows}
\usepackage{lscape}

\newtheorem{theorem}{Theorem}[section]

\newtheorem{assu}{Assumption}[section]
\newtheorem{proof}{Proof}[section]
\newtheorem{remark}{Remark}
\newtheorem{definition}{Definition}[section]

\biboptions{authoryear}

\begin{document}

\begin{frontmatter}

\title{Data Driven Robust Estimation Methods for Fixed Effects Panel Data Models}

\author[first]{Beste Hamiye Beyaztas\corref{correspondingauthor}}
\cortext[correspondingauthor]{Corresponding author}

\author[second]{Soutir Bandyopadhyay}

\address[first]{Department of Statistics, Istanbul Medeniyet University, Istanbul, Turkey}
\address[second]{Department of Applied Mathematics, Statistics Colorado School of Mines, Golden, Colorado, USA}

\begin{abstract}
The panel data regression models have gained increasing attention in different areas of research including but not limited to econometrics, environmental sciences, epidemiology, behavioral and social sciences. However, the presence of outlying observations in panel data may often lead to biased and inefficient estimates of the model parameters resulting in unreliable inferences when the least squares (LS) method is applied. We propose extensions of the M-estimation approach with a data-driven selection of tuning parameters to achieve desirable level of robustness against outliers without loss of estimation efficiency. The consistency and asymptotic normality of the proposed estimators have also been proved under some mild regularity conditions. The finite sample properties of the existing and proposed robust estimators have been examined through an extensive simulation study and an application to macroeconomic data. Our findings reveal that the proposed methods often exhibits improved estimation and prediction performances in the presence of outliers and are consistent with the traditional LS method when there is no contamination.
\end{abstract}


\end{frontmatter}

\section{Introduction} \label{Sec:1}

Panel data refers to the two-dimensional data in which cross-sectional units are observed over time. Its grouping structure allows to reflect the nested phenomena so that the characteristics of cross-sectional units are entrenched over time and vice-versa (see, \cite{Bickel2007}). Over last several years, its increasing availability, demanding methodology and better ability to model the complexity of human behavior than a pure cross-section or time series data are the primary reasons behind the excessive growth in its study  (see,\cite{Hsiao2007}). For a comprehensive account of the technical details on linear panel data models, please refer to \cite{Baltagi2005}, \cite{Fitzmaurice2004}, \cite{Greene2003}, \cite{Maddala1973}, \cite{Mundlak1978}, \cite{Diggle2002}, \cite{Hill2007}, \cite{Wallace1969}, \cite{Wooldridge2002}, and the references therein. 

In recent years, the panel data regression models have received increasing attention in the research of different fields such as econometrics and biostatistics, since these models allow the unobserved individual-specific heterogeneity to be taken into account (cf. \cite{Baltagi2005} and \cite{Hsiao1985} for more details). The statistical appeal of panel data models relies on the fact that these focus particularly on explaining within variations over time and provide controls over individual heterogeneity. The most widely used regression techniques for making statistical inferences regarding the parameters of linear panel data regression models are typically based on the LS approaches. However, traditional estimation techniques require a number of quite restrictive and unrealistic assumptions such as the normality of error distribution, strict exogeneity with respect to the error terms and homoscedasticity of the error terms (cf. \cite{Kutner2004}, \cite{Greene2017}). Also, the panel data may include various types of outliers such as  either vertical, horizontal or leverage as noted in \cite{Rousseeuw1990} and \cite{Bakar2015}. Furthermore, the outlying observations may be concentrated in some blocks such that the fraction of contaminated values per one cross-sectional unit constitutes at least a half of the observations over time periods (cf. \cite{Bramati2007} and \cite{Aquaro2013}). Hence, the classical ordinary least squares (OLS) based methods may considerably be affected in the presence of data contamination and outliers caused by the measurement error, typing error, transmission/copying error and naturally unusual observations as noted in \cite{Rousseeuw2003}, \cite{Maronna2006} and \cite{Bakar2015}. As a result of this, the well-known estimators, such as within group LS estimators used in panel data models with fixed effects, often lead to unreliable estimates of the model parameters. In addressing the problem, more robust alternatives to LS methods having a high breakdown point (BP), such as Least Trimmed Squares (LTS) and S-estimators have been introduced by \cite{Rousseeuw1984} and \cite{RousseeuwYohai1984} for linear regression models. One of the various measures characterizing the robustness of the estimator is the breakdown point which measures the minimum proportion of the data that can significantly change the estimates (cf. \cite{Genton2003}, \cite{Davies2005} and \cite{Aquaro2013}). As pointed out by \cite{Rousseeuw1987}, the asymptotic breakdown point of the LS estimator is zero in the presence of contaminated data sets. Hence, for contaminated data sets, the importance using robust and positive breakdown point methods in estimating model parameters has been emphasized by many authors; see, for instance, \cite{Hampel1986}, \cite{Simpson1992}, \cite{Wagenvoort2002}, \cite{Maronna2006}, \cite{Cizek2008}. This is more crucial for panel data since the outlying data points may be masked and not be directly detectable using standard outlier diagnostics due to the complex structure of the data. 

In spite of the fact that building up the robust methods have been well-studied in estimation of the linear regression parameters, the number of available approaches on the robust methods for static panel data models is fairly limited (cf. \cite{Bramati2007},\cite{Aquaro2013}). A few different approaches within the robust estimation framework for the fixed effects panel data models have recently been proposed by \cite{Bramati2007}, \cite{Namur2011}, \cite{Aquaro2013}, \cite{Bakar2015}, \cite{Visek2015} and \cite{Midi2018}. \cite{Bramati2007} have proposed the robust alternatives to the within group LS estimator with positive breakdown point by extending well-known robust regression estimators, such as LTS estimator of\cite{Rousseeuw1984} and a combination of M and S estimates, namely, MS estimates of \cite{Maronna2000}. Another robust estimation approach has been proposed in \cite{Aquaro2013} based on two different data transformations (i.e. first-difference and pairwise-difference transformation) by applying the efficient weighted least squares estimator of \cite{Gervini2002} and the reweighted LTS estimator of \cite{Cizek2010} in the fixed effects linear panel data framework. \cite{Visek2015} has developed a robust algorithm by weighting down the large order statistics of squared residuals. Moreover, \cite{Bakar2015} have considered a robust centering method by employing the MM-centering procedure to the data. Then, the authors use the within group Generalized M Estimator (WGM) of \cite{Bramati2007} to estimate the parameters of fixed effect panel data model. More recently, a Weighted Least Square (WLS) procedure based on MM-centering method, which is highly resistant in the presence of leverage points and vertical outliers, has been proposed by \cite{Midi2018}. 

In this study, we concentrate on a class of robust estimators, i.e. M-estimators, instead of investigating different data transformations for fixed effects panel data model. A dispersion function (also called the \textit{loss function}), that varies at large values more slowly in comparison to the squared function of the residuals, is attempted to be minimized in M-estimation approaches. However, the robustness against outliers is achieved by some efficiency loss when unnecessarily resistance occurs. (e.g., \cite{Hampel1986}, \cite{Lindsay1994}, \cite{Wang2007} and \cite{Jiang2019}). Hence, it is critical to choose a dispersion function with an appropriate resistance level in obtaining efficient estimates of the parameters as noted in \cite{Wang2007}. To determine the necessary level of robustness, a tuning parameter (also called the regularization parameter) in the dispersion function requires to be selected appropriately based on the possible proportion of outliers (please, see \cite{Wang2007} and \cite{Wang2018} for more details). For some dispersion functions, such as Tukey's bisquare and Huber's functions, a data-driven procedure that automatically chooses the value of the tuning parameter has been introduced by \cite{Wang2007} and \cite{Wang2018} in the context of regression models. The main idea behind this method is to achieve desirable level of robustness against outliers without sacrificing estimation efficiency. Also, a data-dependent procedure to choose the value of tuning parameter in Exponential loss function has been proposed by \cite{Wang2013} in obtaining penalized robust estimators with maximum efficiency and asymptotic breakdown point $1/2$. In this paper, we adapt the M-estimation techniques with automatic selection of tuning parameter introduced by \cite{Wang2007} and \cite{Wang2013} to fixed effects panel data models to achieve resistant estimation against outliers as well as improving estimation efficiency. To construct robust alternative procedures based on Huber's, Tukey's bisquare and Exponential loss functions to the within group LS estimator, we use within group transformation on mean centered data to eliminate the individual effects. Thus, we do not provide the detailed explanations of the robust methods based on the different data transformations proposed by \cite{Aquaro2013}. The asymptotic properties of the proposed M-estimation methods are investigated by establishing the consistency and asymptotic normality in the context of fixed effects linear panel data models. Also, to investigate the finite sample performances of the proposed robust estimators, Monte Carlo experiments are carried out under various contamination schemes and two levels of contamination. Furthermore, the effects of different types of outliers including vertical outliers and leverage points on the proposed estimation procedures have been examined. The numerical results demonstrate that the proposed M-estimators based on data-driven procedures yield more accurate and precise estimates of the model parameters compared to within group LS estimator and within group MS estimator (WMS) of \cite{Bramati2007} for the increasing level of contamination, in general. Moreover, the predictive ability of the models constructed by estimating coefficients using our proposed methods is better than that of WMS method.

The remainder of the paper is organized as follows. In Section~\ref{Sec:2}, we present a detailed information on the fixed effects linear panel data models and within group LS estimator, followed by a discussion on the existing robust estimation procedures in Section~\ref{Sec:3}. In Section~\ref{Sec:4}, we describe our method to obtain the M-estimators based on Tukey's bisquare, Huber's and Exponential loss functions with a data-dependent tuning in the context of linear fixed effects panel data models. In Section~\ref{Sec:5}, we present the large sample properties of the proposed estimators. The finite sample performances of the proposed M-estimators are provided by an extensive simulation study and the results are compared with the robust WMS and traditional LS methods in Section~\ref{Sec:6}. In Section~\ref{Sec:7}, we apply proposed robust methods to real macroeconomic data. Finally, Section~\ref{Sec:8} concludes the work with a few remarks.

\section{Fixed Effects Panel Data Models} \label{Sec:2}
A linear fixed-effects panel data model with a random sample $\left\lbrace \left( y_{it}, x_{it}, \alpha_i \right), i = 1, \cdots, N, t = 1, \cdots, T \right\rbrace$ can be represented as
\begin{equation} \nonumber
y_{it} = x_{it}^T \beta + \alpha_i + \varepsilon_{it}, 
\end{equation}
where $y_{it}$ denotes the response variable, the $K$-dimensional random explanatory variables are denoted by $x_{it}$'s and $\beta \in \mathbb{R}^K$ are the vector of regression parameters. The subscript $i$ denotes the individuals observed over time periods $t$. Finally, the $\alpha_i$'s and $\varepsilon_{it}$'s respectively represent the unobservable individual specific effects and the independent and identically distributed (i.i.d.) error terms with $E \left( \varepsilon_{it} \vert x_{i1}, \cdots, x_{iT}, \alpha_i \right) = 0$, $E \left( \varepsilon_{it}^2 \vert x_i, \alpha_i \right) = \sigma_{\varepsilon}^2 I_T$ where $I_T$ is the identity matrix and $E \left( \varepsilon_{it} \varepsilon_{is} \vert x_i, \alpha_i \right) = 0$ for $t \neq s$. For simplifying notation, the panel data regression model can be expressed in more compact form, by stacking observations over time and individuals, as in Equation~(\ref{Eq:1}) given below.
\begin{equation} \label{Eq:1}
y = \alpha \otimes e_T + \mathbf{X} \beta + \varepsilon,
\end{equation}
where $y = \left( y_{1}, \cdots, y_{N} \right)^T$ and $\mathbf{X} = \left( x_1, \cdots, x_N \right)^T$, respectively, are  an $NT \times 1$ vector with $y_i = \left( y_{i1}, \cdots, y_{iT} \right)$ and an $NT \times K$ matrix of regressors with $\mathbf{X}_i = \left( x_{i1}^T, \cdots, x_{iT}^T \right)^T$, $\alpha$ is an $N \times 1$ vector consisting of the individual effects $\alpha_i$ for $i = 1 \cdots N$, $e_T$ is a $T \times 1$ vector of ones and $\otimes$ denotes the Kronecker product. 

An important advantage of the fixed effects model is the elimination of the individual effects from the model equation when estimating the parameter vector, $\beta$. The within group transformation removes the fixed effects by using the time averages of $y_{it}$, $x_{it}$ and $\varepsilon_{it}$ for each-cross sectional unit: $\bar{y}_{i.} = T^{-1} \sum_{t = 1}^T y_{it}$, $\bar{x}_{i.} = T^{-1} \sum_{t = 1}^T x_{it}$ and $\bar{\varepsilon}_{i.} = T^{-1} \sum_{t = 1}^T \varepsilon_{it}$. Then, the within group transformed model for the mean-centered data is obtained as follows.
\begin{equation} \nonumber
\ddot{y}_{it} = \ddot{x}_{it}^T \beta + \ddot{\varepsilon}_{it}
\end{equation}
where $\ddot{y}_{it} = y_{it} - \bar{y}_{i.}$, $\ddot{x}_{it} = x_{it} - \bar{x}_{i.}$ and $\ddot{\varepsilon}_{it} = \varepsilon_{it} - \bar{\varepsilon}_{i.}$. Under the assumptions of fixed effects models, the within group LS estimator of $\beta$, $\hat{\beta}_{\left( LS \right)}$ is obtained by running regression of $\ddot{y}_{it}$ on $\ddot{x}_{it}$ using OLS method, as follows.
\begin{equation} \nonumber
\hat{\beta}_{\left( LS \right)} = \left( \sum_{i=1}^N \sum_{t=1}^T \ddot{x}_{it}^T \ddot{x}_{it} \right)^{-1} \left( \sum_{i=1}^N \sum_{t=1}^T \ddot{x}_{it}^T \ddot{y}_{it} \right)
\end{equation}

The within group estimator fulfills three equivariance properties with respect to scale, regression and affine transformations since it is linear. Suppose $R \left( \lbrace x_{it}, y_{it} \rbrace \right)$ denotes a panel regression estimator as a function of data. 

\begin{definition} \label{def:1}
If, for any constant $c \in \mathbb{R}$, 
\begin{equation} \nonumber
R \left( \lbrace x_{it}, cy_{it} \rbrace \right) = c R \left( \lbrace x_{it}, y_{it} \rbrace \right)
\end{equation}
the estimator $R$ is scale equivariant.
\end{definition}
\begin{definition} \label{def:2}
If, for any $K \times 1$ vector of constants $\nu$, 
\begin{equation} \nonumber
R \left( \lbrace x_{it}, y_{it} + x_{it}^T \nu \rbrace \right) = R \left( \lbrace x_{it}, y_{it} \rbrace \right) + \nu
\end{equation}
it is regression equivariant.
\end{definition}
\begin{definition} \label{def:3}
If it follows for non-singular $A \in \mathbb{R}^{K \times K}$
\begin{equation} \nonumber
R \left( \lbrace A^T x_{it}, y_{it} \rbrace \right) = A^{-1} R \left( \lbrace x_{it}, y_{it} \rbrace \right)
\end{equation}
then, this estimator satisfies the affine equivariance property.
\end{definition}

The within group LS estimator defined above is known to be highly sensitive to the presence of outliers and aberrant observations. The motivation of this work is to construct robust estimation procedures which are less sensitive in the presence of outliers and erroneous observations. In this paper, we propose extensions of the data-dependent approaches proposed by \cite{Wang2007}, \cite{Wang2013} and \cite{Wang2018} for obtaining robust and more efficient estimates in fixed effects panel data models. This study considers three approaches: the first one is based on exponential squared loss (ESL) function suggested by \cite{Wang2013} to improve the robustness of variable selection procedures in context of penalized regression methods. Also, they propose a class of penalized robust regression estimators based on ESL and a data-driven procedure, which allows to select a tuning parameter depending on the proportion of outliers in the data. These procedures yield highly robust and also, efficient estimates by achieving the highest asymptotic breakdown point of $1/2$. In other approaches, we propose to use Huber's function and Tukey's bisquare function with data-dependent regularization parameters to obtain  more efficient estimates within the fixed effects models framework. \cite{Wang2007} and \cite{Wang2018} have proposed a data-driven method for the automatic selection of a tuning constant that can be applied to the various dispersion functions such as the Huber, Tukey's bisquare and ESL functions for regression models. They obtain considerably better results by improving the efficiency in estimating the regression parameters. Next, we briefly discuss the existing robust estimators and describe our proposed method to estimate the fixed effects panel data models.

\section{Robust Estimation for Fixed Effects Panel Data Models} \label{Sec:3}

The existing literature regarding the robustness of estimators for static fixed effect panel data models is fairly limited. For the purpose of constructing highly robust procedures, \cite{Bramati2007} propose the WGM estimator and the WMS estimator, which have asymptotically breakdown point of $1/4$. In the first approach, \cite{Bramati2007} suggest robustly centring the variables by using the within group medians (as described in Equation~(\ref{Eq:2})).
\begin{eqnarray} \label{Eq:2}
\tilde{y}_{it} &=& y_{it} - \underset{t}{\mathrm{median}}\left(y_{it}\right) \\ \nonumber
\tilde{x}_{it} &=& x_{it} - \underset{t}{\mathrm{median}}\left(x_{it}\right)
\end{eqnarray} 
Next, LTS regression is suggested to employ on the centered data to obtain initial estimates. Finally, a weighted LS estimation, which uses Tukey's bisquare function with the fixed tuning parameter $c=4.685$ and a multivariate S-estimator for down-weighting of leverage points, is performed to construct the WGM estimator. Although the WGM estimator can achieve the breakdown point up to $1/4$, it has crucial limitations of not being regression and affine equivariant because of the non-linearity of median transformation.

To construct the WMS estimator, the fixed effects are first estimated by considering them as regression coefficients. Then, the MS regression estimator of \cite{Maronna2000} can be implemented to the panel data, which uses M-estimators and S-estimators for the discrete explanatory variables and the continuous ones, respectively. It has been noted that the WMS, which has the advantages of being regression and affine equivariant, and the WGM procedure yield quite similar efficient estimation asymptotically. Thus, we compare our proposed procedures with the WMS estimator by considering its superiority over the WGM estimator.

Recently, \cite{Aquaro2013} propose robust estimation procedures based on the pairwise difference transformations by applying the efficient weighted LS estimator (REWLS) of \cite{Gervini2002} and the reweighted LTS (RLTS) estimator of \cite{Cizek2010}. Since the novelty of their approaches is to use different data transformations, the authors demonstrate the equivariance, robust, and asymptotic properties of the proposed estimators. The finite sample performances of the WGM and WMS estimators of \cite{Bramati2007} and the REWLS and RLTS estimators of \cite{Aquaro2013} are similar and they have equal breakdown points according to a given data transformation, see \cite{Aquaro2013}. Hence, the numerical results discussed in Section \ref{Sec:5} do not include the results related to the performances of the methods proposed by \cite{Aquaro2013}.

\section{New Robust Estimators Using Well-Known Loss Functions with Automatic Selection of Tuning Parameter} \label{Sec:4}

A broad class of robust estimation methods is the M-estimation approach which is based on minimizing dispersion function that more slowly varies at large values compared to the squared function of the residuals (see \cite{Wang2007} for details). In the M-estimation approach, the robustness is achieved by sacrificing some of efficiency in presence of unnecessarily excess resistance (e.g., \cite{Hampel1986}, \cite{Lindsay1994}, \cite{Wang2007} and \cite{Jiang2019}). To obtain necessary degree of robustness, a tuning parameter in the dispersion function requires to be chosen appropriately depending on the possible proportion of outliers, as noted in \cite{Wang2007} and \cite{Wang2018}. To this end, for a specific family of dispersion functions, a data-driven approach that automatically selects the value of the tuning constant has been introduced by \cite{Wang2007} in the context of regression models. The main idea underlying this method is to achieve desirable robustness level in presence of outliers without loss of estimation efficiency. Moreover, \cite{Wang2018} extend this procedure for linear regression models with autoregressive errors in analysing water quality data.

In this study, we focus on the robust estimation approaches based on loss functions with automatic selection of tuning parameter proposed by \cite{Wang2007} and \cite{Wang2013} to achieve resistant estimation against outliers and improve estimation efficiency in fixed effects panel data models. To construct alternative procedures to the methods discussed in Section \ref{Sec:3}, we use within group transformation for mean centered data to eliminate the individual effects. Thus, we do not provide the numerical results and the detailed explanations of the robust methods based on the different data transformations proposed by \cite{Aquaro2013}.

The M-estimation approach relies on minimization of a loss function, $\rho(\cdot)$, instead of minimizing sum of squared errors since the LS method is highly sensitive to the distortions caused by outliers (\cite{Rousseeuw2003}). The Huber and Tukey's bisquare functions, which are the most commonly used loss functions in robust regression, are defined as in Equations~(\ref{Eq:3})-(\ref{Eq:4}), respectively. (see, \cite{Wang2018}) 
\begin{itemize}
\item Huber's function 
\begin{equation} \label{Eq:3}
\rho(u) = \left \{\begin{array}{ll}
      \frac{u^2}{2} & \text{if}\;\lvert u \rvert \leqslant c \\
      c\lvert u \rvert - \frac{c^2}{2} & \text{if}\;\lvert u \rvert > c
    \end{array}
  \right.
\end{equation} 
\item Tukey's bisquare function 
\begin{equation} \label{Eq:4}
\rho(u) = \left \{\begin{array}{ll}
      1 - \left\lbrace 1 - \left( \frac{u}{c} \right)^2 \right\rbrace^3 & \text{if}\; \lvert u \rvert \leqslant c\\
      1 & \text{if}\;\lvert u \rvert > c
    \end{array}
  \right.
\end{equation} 
\end{itemize}
Also, \cite{Wang2013} propose to use ESL function to obtain a class of penalized robust regression estimators with asymptotic breakdown point of $1/2$ for variable selection, and it can be expressed as follows. 
\begin{itemize}
\item ESL function
\begin{equation} \nonumber
\rho(u) = 1 - exp \left\lbrace - \left( \frac{u^2}{c} \right) \right\rbrace
\end{equation} 
\end{itemize}
Then, the robust estimator of $\beta$ is obtained as the solution of the following estimating functions by minimizing the loss functions $\sum_{i=1}^{N} \sum_{t=1}^{T} \rho \left\lbrace \frac{\left( y_{it} - x_{it}^T \beta - \alpha_i \right)}{\hat{\sigma}} \right\rbrace$ for a given estimate of the scale parameter, $\hat{\sigma}$ 
\begin{equation} \nonumber
U(\beta) = \sum_{i=1}^{N} \sum_{t=1}^{T} x_{it} \psi \left\lbrace \frac{\left( y_{it} - x_{it}^T \beta - \alpha_i \right)}{\hat{\sigma}} \right\rbrace = \underset{K \times 1}{\left( \mathbf{0} \right)}
\end{equation}
where $\psi$ denotes the sub-gradient function of $\rho\left( \right)$. It can be expressed as in Equations~(\ref{Eq:5})-(\ref{Eq:7}), for Huber's, Tukey's bisquare and ESL functions, respectively. 
\begin{itemize}
\item Huber's function
\begin{equation} \label{Eq:5}
\psi(u) = \left \{\begin{array}{ll}
      u & \text{if}\;\lvert u \rvert \leqslant c\\
      sign\left( u \right) c & \text{if}\; \lvert u \rvert > c
    \end{array}
  \right.
\end{equation} 
\item Tukey's bisquare
\begin{equation} \label{Eq:6}
\psi(u) = \left \{\begin{array}{ll}
      u \left\lbrace 1 - \left( \frac{u}{c} \right)^2 \right\rbrace^2 & \text{if}\;\lvert u \rvert \leqslant c\\
      0 & \text{if}\;\lvert u \rvert > c
    \end{array}
  \right.
\end{equation} 
\item The ESL function
\begin{equation} \label{Eq:7}
\psi(u) = u exp\left\lbrace - \left( \frac{u^2}{c} \right) \right\rbrace
\end{equation} 
\end{itemize}

Let us consider the residuals $e_{it} = \ddot{y}_{it} - \ddot{x}_{it}^T \hat{\beta}_{\left( LS \right)}$ when OLS regression is performed on the within group transformed data. Then, the estimating equation can be rearranged as follows 
\begin{eqnarray} \label{Eq:8} 
U(\beta) &=& \sum_{i=1}^{N} \sum_{t=1}^{T} \ddot{x}_{it} \psi \left\lbrace \frac{\left( \ddot{y}_{it} - \ddot{x}_{it}^T \beta \right)}{\hat{\sigma}} \right\rbrace \\ 
&=& \sum_{i=1}^{N} \sum_{t=1}^{T} \omega_{it} \ddot{x}_{it} \overline{e}_{it} \nonumber
\end{eqnarray}
where $\omega_{it} = \psi \left( \overline{e}_{it} \right)/\overline{e}_{it}$ denote the weights obtained by the weighting function $W(e) = \rho^{\prime}(e)/e$ and $\overline{e}_{it} = \left( \ddot{y}_{it} - \ddot{x}_{it}^T \beta \right) / \hat{\sigma}$. The $U$ defined in Equation~(\ref{Eq:8}) has a form of score functions with the weights, $\omega_{it}$. Thus, our proposed robust estimators can be obtained as follows
\begin{equation} \nonumber
\hat{\beta}_M = \left( \sum_{i=1}^N \sum_{t=1}^T \ddot{x}_{it}^T \omega_{it} \ddot{x}_{it} \right)^{-1} \left( \sum_{i=1}^N \sum_{t=1}^T \ddot{x}_{it}^T \omega_{it} \ddot{y}_{it} \right).
\end{equation}

Since $W$ is a function of $\beta$ and $\sigma$, an iterative procedure has been followed at the $k$-th iteration based on the previous $\hat{\beta}^{\left(k-1\right)}$ to determine the weights, $\hat{\omega}_{it} = \omega_{it} \rvert_{\beta = \hat{\beta}^{\left(k-1\right)}}$, . The M-estimation method assumes that $E \psi \left( \varepsilon_i \right) = 0$ and $E \psi^2 \left( \varepsilon_i \right) = \sigma_{\psi}^2$. Let's also suppose that
\begin{equation} \nonumber
b = \frac{\partial E \psi \left( \varepsilon_{it} + \delta \right)}{\partial \delta} \biggr\rvert_{\delta = 0}. \nonumber
\end{equation}
Under the finite variance assumption with some regularity conditions defined in detail in Section \ref{Sec:5} and $E \lvert \psi \left( \varepsilon_i + \delta \right) - \psi \left( \varepsilon_i \right) \rvert^2 = o(1)$  as $\delta \rightarrow 0$, the consistent estimator $\hat{\beta}_M$ is obtained by solving $U(\beta) = 0$ and it follows that
\begin{equation} \nonumber
\sqrt{N}\left( \hat{\beta}_M - \beta \right) \sim N \left(0, \tau^{-1} \sigma^2_{\varepsilon} \left[ E \left( \ddot{\mathbf{X}}_{i}^T \ddot{\mathbf{X}}_{i} \right) \right]^{-1} \right),
\end{equation}
where $E \left( \ddot{\mathbf{X}}_{i}^T \ddot{\mathbf{X}}_{i} \right)$ is estimated by $N^{-1} \sum_{i=1}^N \ddot{\mathbf{X}}_{i}^T \ddot{\mathbf{X}}_{i}$  (or $N^{-1} \sum_{i=1}^N \sum_{t=1}^T \ddot{x}_{it}^T \ddot{x}_{it}$) and $\tau = b^2 / \sigma_{\psi}^2$ denotes the scalar efficiency factor. The subgradient function $\psi$, which has large values of the efficiency factor $\tau$, lead to some gain in efficiency of the estimator $\beta$ as noted in \cite{Wang2007}. Thus, the loss function $\rho \left( \cdot \right)$ with the largest value of $\tau$ should be chosen to obtain more efficient estimates of the parameters. 

As noted in \cite{Bramati2007} and \cite{Wang2007}, the choice of constant, $c$ may have a great impact to achieve a good trade-off between efficiency and robustness degree. For example, in constructing the WGM estimator proposed by \cite{Bramati2007}, the value of $c$ is chosen as $4.685$ where Tukey's bisquare function is used. For Huber's function, its default value in the \texttt{R} package \texttt{rlm} equals to $1.345$ to obtain about 90\% efficiency when the data follow Normal distribution, see \cite{Wang2007}. In traditional robust procedures, the value of $c$ for any loss function must be predetermined by taking into account the desirable level of robustness. However, the efficiency of the estimators varies significantly with the different choices of $c$. Therefore, the value of tuning parameter $c$ requires to be chosen depending on the proportion of the outlying points in the data or distribution of the data to maximize the estimation efficiency. Thus, in this study, best value of the tuning parameter $c$ is chosen by maximizing the value of $\tau$, and the nonparametric estimate of $\tau$, proposed by \cite{Wang2018}, is calculated by 
\begin{equation} \nonumber
\hat{\tau} \left( c \right) = \frac{\left\lbrace \sum_{i=1}^{N} \sum_{t=1}^{T} \psi^{\prime} \left( \hat{e}_{it} \right) \right\rbrace^2}{NT \sum_{i=1}^{N} \sum_{t=1}^{T}  \psi^2 \left( \hat{e}_{it} \right)} \nonumber
\end{equation}
where $\hat{e}_{it} = \frac{(\ddot{y}_{it} - \ddot{x}_{it}^T \hat{\beta})}{\hat{\sigma}}$, $\hat{\beta}$ and $\hat{\sigma}$ denote the current estimates of $\beta$ and $\sigma$, respectively.

Based on the above, we present the data-driven procedure introduced by \cite{Wang2007} to obtain the proposed robust estimators of $\beta$ in panel data regression models when Huber's and Tukey's bisquare functions are used as loss functions.
\begin{itemize}
\item[Step 1.] Compute the within group LS estimate, $\hat{\beta}_{\left( LS \right)}$ of the parameter vector, $\beta$ for the panel data regression model.
\item[Step 2.] Calculate the residuals $e_{it} = \ddot{y}_{it} - \ddot{x}_{it}^T \hat{\beta}_{\left( LS \right)}$ and obtain the initial estimate of $\sigma$ with the following equation.
\begin{equation} \nonumber
\hat{\sigma} = Median \left\lbrace \lvert \ddot{y}_{it} - \ddot{x}_{it}^T \hat{\beta}_{\left( LS \right)} \rvert \right\rbrace /0.6745
\end{equation}
\item[Step 3.] Obtain the $c$ value (within a range of values of $0$ to $3$ for Huber's function and $1$ to $10$ for Tukey's bisquare function as suggested by \cite{Wang2018}) which has the largest efficiency factor $\hat{\tau} \left( c \right)$.
\item[Step 4.] Finally calculate the robust estimator of $\beta$ by using the Huber's function and Tukey's bisquare function with $c = \hat{c}$ obtained in the previous step.
\end{itemize}

Next, the complete algorithm to determine the value of tuning parameter $c$ in the ESL function is provided by extending the data-driven procedure for penalized regression of \cite{Wang2013} to the linear panel data models with fixed effects. For more detailed information, please see \cite{Wang2013}.
\begin{itemize}
\item[Step 1.] Let $\mathbf{Z} = \lbrace \ddot{x}_{it}, \ddot{y}_{it} \rbrace_{i=1, t=1}^{N, T}$ be a random sample that contains $m$ bad points and $NT-m$ good points  denoted by $\mathbf{Z}_m = \lbrace Z_1, \cdots, Z_m \rbrace$ and $\mathbf{Z}_{NT-m} = \lbrace Z_{m+1}, \cdots, Z_{NT} \rbrace$, respectively. Then, obtain the residuals $e_{it} (\hat{\beta}_{0})  = \ddot{y}_{it} - \ddot{x}_{it}^T \hat{\beta}_{0}$ using the initial high breakdown coefficient $\hat{\beta}_{0}$ for $i = 1, \cdots, N$, $t = 1, \cdots, T$ and calculate median absolute deviation estimator (MAD), $\hat{\sigma}_{MAD} = 1.4826 \times median \lvert e_{it} (\hat{\beta}_{0}) - \underset{t} median (e_{it}(\hat{\beta}_{0})) \rvert$. Then, generate pseudo outlier set $\mathbf{Z}_m = \lbrace (\ddot{x}_{it}, \ddot{y}_{it}): \lvert e_{it} (\hat{\beta}_{0}) \rvert \geq 2.5 (\hat{\sigma}_{MAD}) \rbrace$ where $m = \# \lbrace 1 \leq i \leq N, 1 \leq t \leq T : \lvert e_{it} (\hat{\beta}_{0}) \rvert \geq 2.5 (\hat{\sigma}_{MAD}) \rbrace$ and $\mathbf{Z}_{NT-m} = \mathbf{Z}/\mathbf{Z}_m$.
\item[Step 2.] Let $c_N$ denote the minimizer of det$\left( \hat{V} \left( c \right) \right)$ in the set $G=\left\lbrace c: \xi \left( c \right)\in \left( 0, 1 \right] \right\rbrace$ where
\begin{equation} \label{Eq:9}
\xi \left( c_N \right) = \frac{2m}{NT} + \frac{2}{NT} \sum_{t=m+1}^T \rho_{c_N} \left\lbrace e_{it} \left( \hat{\beta}_0 \right) \right\rbrace
\end{equation}
for $i = 1, \cdots, N$  for a contaminated sample $\mathbf{Z}$ and det$\left(\cdot\right)$ denotes the determinant operator, $\hat{V} \left( c \right) = \left\lbrace \hat{I} \left( \hat{\beta}_0 \right) \right\rbrace^{-1} \tilde{\Sigma} \left\lbrace \hat{I} \left( \hat{\beta}_0 \right) \right\rbrace^{-1}$, and
\begin{eqnarray} \nonumber
\hat{I} \left( \hat{\beta}_0 \right) &=& \frac{2}{c} \left\lbrace \frac{1}{NT} \sum_{i=1}^N \sum_{t=1}^T exp\left( - \frac{e_{it}^2 \left( \hat{\beta}_0 \right)}{c}\right) \left( \frac{2e_{it}^2 \left( \hat{\beta}_0 \right)}{c}-1\right) \right\rbrace \\ \nonumber
&\times& \left( \frac{1}{NT}\sum_{i=1}^N \sum_{t=1}^T \ddot{x}_{it} \ddot{x}_{it}^T \right), \\ \nonumber
\tilde{\Sigma} &=& cov \left\lbrace exp\left( - \frac{e_{it}^2 \left( \hat{\beta}_0 \right)}{c}\right) \frac{2e_{it} \left( \hat{\beta}_0 \right)}{c} \ddot{x}_{1t}, \cdots, exp\left( - \frac{e_{it}^2 \left( \hat{\beta}_0 \right)}{c}\right) \frac{2e_{it} \left( \hat{\beta}_0 \right)}{c} \ddot{x}_{Nt} \right\rbrace_{t=1}^T. \nonumber
\end{eqnarray}
\item[Step 3.] Update $\hat{\beta}_0$ based on the selected value of $c_N$ in Step 2.
\end{itemize}
It should be noted that we use the MM-estimator of \cite{Yohai1987} to obtain the inital estimate $\hat{\beta}_{0}$ when detecting the outliers in Step 1 in the algorithm for ESL function. This provides to compute $\xi \left( c_N \right)$ by using Equation~(\ref{Eq:9}). Then, the Steps 1-3 are repeated until $\hat{\beta}_0$ and $c_N$ converge. To achieve high efficiency, the tuning parameter $c_N$ is selected by minimizing the determinant of asymptotic covariance matrix as in Step 2 as noted in \cite{Wang2013}. In the last Step, to obtain $\hat{\beta}_M$, $\hat{\beta}_0$ is updated in Step 3 and the algorithm is repeated until convergence. \cite{Wang2013} have emphasized that for the convergence of their computing algorithm, it is enough to repeat Steps 1-3 once. If the initial estimator $\hat{\beta}_0$ is robust having asymptotic breakdown point of $1/2$ and $c_N$ is chosen such that $\xi \left( c_N \right)\in \left( 0, 1 \right]$, then the breakdown point of $\hat{\beta}_M$, $BP\left(\hat{\beta}_M; \mathbf{Z}_{NT-m}, c_N \right)$ is asymptotically $1/2$, see \cite{Wang2013}, Theorem 2. This leads us to choose $c_N$ to achieve the highest efficiency. 

\section{Asymptotic Properties} \label{Sec:5}
In this section, we established the asymptotic properties of the regression estimators of under fixed effects linear panel data models. Before introducing the asymptotic results, we begin with listing some assumptions.

To obtain the consistent estimates of the parameters of fixed effects linear panel data models, the within group LS method requires to satisfy some assumptions (see \cite{Wooldridge2002}) as listed below:
\begin{assu}
\begin{itemize}
\item[A1.] $E \left( \varepsilon_{it} \vert \mathbf{x}_i, \alpha_i \right)$ for $t= 1, \cdots, T$ implies that the strict exogeneity of $\lbrace x_{it}: t= 1, \cdots, T \rbrace$ conditional on $\alpha_i$'s.
\item[A2.] rank$\left( \sum_{t=1}^T E \left( \ddot{x}_{it}^T \ddot{x}_{it} \right) \right) =$ rank$\left[ E \left( \ddot{\mathbf{X}}_{i}^T \ddot{\mathbf{X}}_{i} \right) \right] = K$ where $\ddot{\mathbf{X}}_{i}$ is a $T \times K$ matrix of predictors.
\item[A3.] $E \left( \varepsilon_{i} \varepsilon_{i}^T \vert \mathbf{x}_i, \alpha_i \right) = \sigma_{\varepsilon}^2 I_T$.
\end{itemize}
\end{assu}
The following objective function is required to be minimized in obtaining our proposed estimators $\hat{\beta}_M$s based on M-estimation techniques, 
\begin{equation} \label{Eq:10}
L \left( \beta \right) = \sum_{i=1}^{N} \rho_{\tau} \left( \frac{\ddot{y}_i - \ddot{\mathbf{X}}_{i}^T \beta}{S_N}\right)
\end{equation}
where $S_N$ denotes a scale parameter, and is a normalized MAD estimator obtained using the residuals $e_{it} = \ddot{y}_{it} - \ddot{x}_{it}^T \hat{\beta}_0$ (cf. \cite{Jiang2019}). It should be noted that the proposed estimators $\hat{\beta}_M$s are regression and scale equivariant since $S_N$ has the advantages of being scale equivariant and invariant to the regression. 

In fixed effects linear panel data regression models, the regularity conditions required for the consistency and asymptotic normality of the proposed estimators are as follows:
\begin{assu} \label{Assumption A}
\begin{itemize}
\item[A1.] $E \left( \ddot{\mathbf{X}}_{i}^T \ddot{\mathbf{X}}_{i} \right)$ is positive definite, and $E \Big\| \mathbf{X}^3 \Big\| < \infty$. 
\item[A2.] $E \psi_{\tau} \left( \varepsilon_i \right) = 0$ and $E \psi_{\tau}^2 \left( \varepsilon_i \right) = \sigma_{\psi_{\tau}}^2$. 
\item[A3.] $E \left[ \psi^{\prime}_{\tau}  \left( \varepsilon / \sigma \right) \right] > 0$ for $\tau$ where prime denotes the derivative.
\end{itemize}
\end{assu}
The assumption $E \psi_{\tau} \left( \varepsilon_i \right) = 0$ is needed to provide Fisher consistent estimating functions $U(\beta)$ as noted in \cite{Wang2007}. 
\begin{theorem} \label{theo:1}
Let Assumption \ref{Assumption A} hold. If $S_N \xrightarrow{\text{p}} \sigma$ as $N \rightarrow \infty$, then a local minimum of $L \left( \beta \right)$ occurs at $\hat{\beta}_M$ such that $\Big\| \hat{\beta}_M-\beta \Big\| = \mathcal{O}_p \left( N^{-1/2} \right)$ where ``$\xrightarrow{\text{p}}$'' represents the convergence in probability.
\end{theorem}

\begin{proof} \label{pro:1}
For any $\nu >0$, a large constant C exists satisfying
\begin{equation} \nonumber
P \left( \inf_{\Big\| u \Big\| = C} L \left( \beta + N^{-1/2}u \right) > L \left( \beta \right) \right) > 1 - \nu
\end{equation}
To prove that the existence of a local minimum of $L \left( \beta \right)$ in the ball $\left\lbrace \beta + N^{-1/2}u: \Big\| u \Big\| \leq C \right\rbrace$ with probability at least $1-\nu$, under A1 and A3 of Assumption \ref{Assumption A}, 
\begin{align*}
&L \left( \beta + N^{-1/2}u \right) - L \left( \beta \right) = \sum_{i=1}^N \rho_{\tau} \left( \frac{\ddot{\mathbf{y}}_{i} - \ddot{\mathbf{X}}_{i}^T \left( \beta + N^{-1/2}u \right)}{S_N} \right) - \sum_{i=1}^N \rho_{\tau} \left( \frac{\ddot{\mathbf{y}}_{i} - \ddot{\mathbf{X}}_{i}^T \beta}{S_N} \right) \\
&=-\frac{1}{\sqrt{N} S_N} \sum_{i=1}^N \psi_{\tau} \left( \frac{\ddot{\mathbf{y}}_{i} - \ddot{\mathbf{X}}_{i}^T \beta}{S_N} \right) \ddot{\mathbf{X}}_{i}^T u + \frac{1}{2 S^2_N} u^T \left[\frac{1}{N} \sum_{i=1}^N \psi_{\tau}^{\prime} \left( \frac{\ddot{\mathbf{y}}_{i} - \ddot{\mathbf{X}}_{i}^T \beta}{S_N} \right)\ddot{\mathbf{X}}_{i}^T \ddot{\mathbf{X}}_{i} + o_p(1) \right] u \\
& \triangleq I_1 + I_2 
\end{align*}
where $\ddot{\mathbf{y}}_{i}$ is a $T \times 1$ vector. Then, by using the fact that $S_N \xrightarrow{\text{p}} \sigma$, the law of large numbers and central limit theorem, the proof of Theorem \ref{theo:1} direclty follows from Theorem 1 of \cite{Jiang2019}.
\end{proof}

\begin{remark}
Theorem \ref{theo:1} guarantees the existence of a consistent estimator under some regularity conditions (cf. \cite{Jiang2019}). It should be noted that the existence of redescending M-estimator is not ensured for the unbounded loss function. Also, \cite{Maronna1981} provide the results for the existence of the redescending M-estimator when some conditions hold.
\end{remark}

\begin{theorem} \label{theo:2}
Under Assumption \ref{Assumption A} and $E \lvert \psi_{\tau} \left( \varepsilon_i + \delta \right) - \psi_{\tau} \left( \varepsilon_i \right) \rvert^2 = o(1)$  as $\delta \rightarrow 0$, if $S_N \xrightarrow{\text{p}} \sigma$ as $N \rightarrow \infty$, then, we obtain
\begin{equation} \nonumber
\sqrt{N}\left( \hat{\beta}_M - \beta \right) \xrightarrow{\text{d}} N \left(0, \frac{E \left[ \psi^2_{\tau}  \left( \varepsilon / \sigma \right) \right]}{\left( E \left[ \psi^{\prime}_{\tau}  \left( \varepsilon / \sigma \right) \right] \right)^2} \sigma^2 \left[ E \left( \ddot{\mathbf{X}}_{i}^T \ddot{\mathbf{X}}_{i} \right) \right]^{-1} \right)
\end{equation}
where ``$\xrightarrow{\text{d}}$'' denotes the convergence in distribution. 
\end{theorem}

\begin{proof} \label{pro:2}
Let assume
\begin{equation*}
\phi_N \left( \theta \right) = \frac{1}{N S_N} \sum_{i=1}^N \psi_{\tau} \left( \frac{\ddot{\mathbf{y}}_{i} - \ddot{\mathbf{X}}_{i}^T \theta}{S_N}\right) \ddot{\mathbf{X}}_{i} 
\end{equation*}
Then, we have the following by Taylor's expansion
\begin{equation*}
\phi_N \left( \hat{\beta}_M \right) = \phi_N \left( \beta \right) + \phi^{\prime}_N \left( \beta \right) \left( \hat{\beta}_M - \beta \right) + \frac{\left( \hat{\beta}_M - \beta \right)^T  \phi_N^{''} \left( \beta^* \right)  \left( \hat{\beta}_M - \beta \right)}{2} 
\end{equation*}
where $\beta^*$ is a vector that lies between $\hat{\beta}_M$ and $\beta$, and also, $\phi^{\prime}_N \left( \cdot \right)$ and $\phi_N^{''} \left( \cdot \right)$ represent the first-order and second-order derivatives of $\phi_N \left( \cdot \right)$, respectively. Then, $\phi_N \left( \hat{\beta}_M \right)=0$ from Equation \ref{Eq:10} (cf. \cite{Jiang2019}). Under A1-A2 of Assumption \ref{Assumption A}, we obtain 

\begin{equation*}
=\frac{S_N}{\sqrt{N}} \sum_{i=1}^N \psi_{\tau} \left( \frac{\ddot{\mathbf{y}}_{i} - \ddot{\mathbf{X}}_{i}^T \beta}{S_N} \right) \ddot{\mathbf{X}}_{i} = \sqrt{N} \left( \hat{\beta}_M - \beta \right) \frac{1}{N} \sum_{i=1}^N \psi_{\tau}^{\prime} \left( \frac{\ddot{\mathbf{y}}_{i} - \ddot{\mathbf{X}}_{i}^T \beta}{S_N} \right)\ddot{\mathbf{X}}_{i}^T \ddot{\mathbf{X}}_{i} + o_p(1) 
\end{equation*}
by Theorem \ref{theo:1}. The proof of Theorem \ref{theo:2} is completed since $S_N \xrightarrow{\text{p}} \sigma$ as $N \rightarrow \infty$.
\end{proof}

\section{Numerical Results} \label{Sec:6}
In this section, we conduct an extensive simulation study to examine the finite-sample properties of proposed and some existing panel data estimators. To investigate the performances of the proposed procedures, three scenarios, namely (i) different sample sizes, (ii) different error distributions, and (iii) various types of outliers, have been considered. The following simulations and all calculations have been carried out using \texttt{R} 3.6.0. on an IntelCore i7 6700HQ 2.6 GHz PC. (The codes can be obtained from the author upon request.)

For the data generation process, we consider the following static fixed-effect panel-data model 
\begin{equation} \label{Eq:dgp}
y_{it} = x_{it}^T \beta + \alpha_i + \varepsilon_{it},~~i = 1, \ldots, N,~~t = 1, \ldots, T ,
\end{equation} 
where $\varepsilon_{it}$'s are assumed to be iid $\text{N}(\mu = 0, \sigma^2 = 1)$ and the vector of parameters is chosen as $\beta^T = \left( \beta_1, \beta_2 \right) = \left( 2.4, -1.2 \right)$. Similar to experimental design of \cite{Aquaro2013}, the unobservable individual effects are generated as follows
\begin{equation*}
\alpha_i = \sum_{t=1}^T x_{it}^T \gamma / \sqrt{T} + \eta_i
\end{equation*}
where $\gamma = \left( 2, 4\right)^T$ and $\eta_i \sim U \left( 0, 12 \right)$ to guarantee homogeneity of the variances of deterministic $\gamma$ and random parts $\eta_i$ of $\alpha_i$. To avoid completely symmetric design as in \cite{Aquaro2013}, the independent variables $x_{itK}$ for $K = 1, 2$ are generated according to 
\begin{displaymath}
x_{itK} \sim \left\{ \begin{array}{ll}
\chi_2^2 - 2 & \textrm{if K=1}\\
N \left( 0, 1 \right) & \textrm{if K=2}, \end{array} \right.
\end{displaymath}
where $\chi_2^2$ denote the chi-squared distribution with 2 degrees of freedom.

The performances of the proposed estimators are evaluated based on $S = 1000$ simulations. In case of the changing sample sizes and presence of outliers, we compute the mean squared errors (MSE): $MSE = \frac{1}{S} \sum_{s=1}^S \Big\| \widehat{\beta}^s-\beta \Big\|^2$, where $\widehat{\beta}^s$, $s=1, \cdots, S$, represents the estimates obtained from $S$ simulated samples. Also, we calculated the squared errors (SE): $SE = \Big\| \widehat{\beta}^s-\beta \Big\|^2$ to examine the effect of the error distributions on the estimators. Note that throughout the experiments, all calculations are performed for the finite sample breakdown points $BP = 0.1$ and $BP = 0.5$ for WMS estimator of \cite{Bramati2007}. However, we present only the results obtained for the choice of $BP = 0.1$ since WMS procedure with this breakdown point produced better results.

\subsection{Sample sizes} \label{Sec:6.1}
In this subsection, different values of cross-sectional dimension, $N = 50, 100, 150, 200, 250$ (by keeping time period fixed at $T = 3$) and time periods, $T = 4, 6, 9, 12, 24$ (for a fixed number of cross-sectional units $N = 50$) are considered to investigate the influence of panel sizes on our proposed estimators. Figure~\ref{Fig:1} illustrates calculated MSE values of the estimators when the errors follow standart normal distribution, $\varepsilon_{it} \sim$ $\text{N}(\mu = 0, \sigma^2 = 1)$. In this figure, the lines representing the performances of the LS and proposed loss functions (Huber, Tukey, Exponential) based estimators using data-dependent regularization parameters overlap as $N \rightarrow \infty$ and $T \rightarrow \infty$ since they have almost same performances under normal errors. However, the WMS procedure are not consistent for fixed time dimension while it is consistent for increasing values of $T$ as noted in \cite{Aquaro2013}. These results also confirm that the proposed methods are asymptotically equivalent to LS estimator as $N$ and/or $T$ increases.

\subsection{Error distributions} \label{Sec:6.2}
The SE values of the estimators are calculated under four different error distributions: $\text{N}(\mu = 0, \sigma^2 = 1)$, Student's t distribution with 5 degrees of freedom ($t_5$), the chi-squared distribution with 4 degrees of freedom $\chi^2_{(4)}$ and standard Cauchy distribution $Cauchy(0, 1)$ to demonstrate the effects of error distributions on the estimation methods. Also, three pairs of cross-sectional sizes and time periods: $(N_1, T_1)=(30,20)$, $(N_2, T_2)=(75, 8)$ and $(N_3, T_3)=(200, 3)$ are considered as in \cite{Aquaro2013}. The simulation results are presented in Figure~\ref{Fig:2}. The skewness of SE values of WMS procedure (with $BP = 0.1$) is greater than the other methods for Normal, $t_5$ and $\chi^2_{(4)}$ distribution of the errors especially for increasing $N$ or decreasing $T$. This means that the variability and mean of SE values obtained for WMS estimator are slightly higher than the other estimators. All the proposed and LS estimators have similar SE values under all the error distributions, except for standard Cauchy distribution. Also, the mean and variance of the SE values for LS method increases more than the other estimators in case of the Cauchy distribution of the error terms while the proposed robust estimators and WMS of \cite{Bramati2007} are less affected.

\subsection{Outliers} \label{Sec:6.3}

This subsection presents the robustness performances of the estimation procedures in case of the various types of outliers. Throughout the simulations, two different levels of cross-sectional sizes and time periods, namely, $N_1 = 120$, $T_1 = 2$ and $N_2 = 80$, $T_2 = 3$, are considered for the fixed panel size consisting of a total of 240 observations. The proportions of contamination are chosen as 5\% and 10\% by determining the number of outliers as $m = 12$ and $m = 24$. Two main types of contamination, namely, random contamination and concentrated (or clustered) contamination as in \cite{Bramati2007} are used to generate contaminated datasets. The outlying points are randomly distributed over all observations for random contamination while the outliers concentrating in some blocks constitute at least a half of the observations within cross-sectional units for concentrated contamination. For more detailed information on the types of contamination and graphs of contamination schemes, please see Figure 1 of \cite{Bramati2007}. The considered contamination schemes depending on the types of outliers are described as follows.

\begin{itemize}
\item[1.] To generate random vertical outliers ($y_{it}^r$), randomly selected original values of the dependent variable are replaced by the observations from an Uniform distribution $y_{it}^r \sim U \left(20,80\right)$.
\item[2.] At first, randomly chosen values of the dependent variable are contaminated following the same rule of first scheme. After that, the values of the independent variables corresponding to the observations contaminated in the dependent variable are replaced by the points coming from a Normal distribution $\text{N}(\mu = 8, \sigma^2 = 4)$ to obtain the random leverage points.
\item[3.] Concentrated vertical outliers are inserted into the data by substituting the random observations from an Uniform distribution $y_{it}^c \sim U \left(79,80\right)$ for the randomly selected blocks of the original values of dependent variable. 
\item[4.] Firstly, the contaminated blocks of the dependent variable are obtained by applying the same rule used in third scheme. In the next step, the concentrated leverage points are inserted into the original blocks of the independent variables corresponding to the blocks already contaminated in dependent variable by replacing them by the random values from a Normal distribution $\text{N}(\mu = 8, \sigma^2 = 4)$.
\end{itemize}

The simulation results are summarized in Tables~\ref{Tab1:outliers-N-120-80-T-2-3-mse}-\ref{Tab2:outliers-rmse}. At first glance, it is conspicuous from Table~\ref{Tab1:outliers-N-120-80-T-2-3-mse} that the LS method leads to obtain distorted estimates of the parameters of interest in the presence of outliers. It can further be seen that LS estimators can be severely degraded when the data is contaminated by the leverage points, compared to the case of vertical outliers. Additionally, the MSE values obtained for LS increase as the level of contamination increases regardless of the types and levels of contaminations and are severely affected in presence of concentrated contamination. Conversely, the proposed estimators (named according to the corresponding loss functions: Huber, Tukey, Exponential) and WMS of \cite{Bramati2007} are highly resistant to outliers. In particular, the proposed robust estimators are quite less affected by the choice of contamination level and/or scheme compared to WMS procedure. One of the remarkable results is that Tukey estimator based on the Tukey's bisquare function produces best results with smallest MSEs among the robust estimators under all contamination schemes for $(N_1,T_1) = (120, 2)$. The performances of the robust estimators can be sorted based on minimizing MSE values as follows: Tukey, Exponential, Huber estimators and WMS for $T_1 = 2$. We further see that, the performance of WMS estimator is sensitive to the increasing level of contamination and concentrated contamination when $T_1 = 2$ as noted in \cite{Aquaro2013} and \cite{Bramati2007}. Although there are no substantial differences between WMS and Exponential loss function based estimator at 5\% contamination for $(N_2,T_2) = (80, 3)$, the performance of WMS method is slightly better than the proposed robust estimators. However, Exponential and Tukey's bisquare functions based estimators produce more accurate and precise estimates of the parameters for random contamination and concentrated contamination, respectively when the level of contamination increases. Our results clearly demonstrate that the all proposed estimators (Huber, Tukey, Exponential) are reasonable competitors that often exhibit improved performance over the traditional LS method and robust WMS estimator especially for increasing level of contamination.   

To confirm the supremacy over the traditional LS and WMS methods, we also calculate the root mean squared errors (RMSE) of predicted values. To this end, the simulated samples are divided into the following two parts: the predictive model is constructed based on training sample to compute the prediction errors from testing sample with cross-sectional size $ N_{test}$. Then, RMSEs of the predicted values are calculated by 
\begin{equation} \label{Eq:rmse}
RMSE = \frac{1}{S} \sum_{s=1}^S \sqrt{\sum_{i=1}^{N_{test}} \sum_{t=1}^{T} (y_{it}^s - \hat{y}_{it}^s)^2 / \left( N_{test} \times T \right)} =\frac{1}{S} \sum_{s=1}^S \sqrt{ \sum_{t=1}^{T_i} (y_{it}^s - \hat{y}_{it}^s)^2 / T_i }
\end{equation}
where $i = 1, 2, \cdots, N_{test}$, $y_{it}^s$ and $\hat{y}_{it}^s$, respectively, denote the observed values and predictions obtained from the predictive model for each of $S$ simulated samples. Throughout simulations, we determine the sizes of test samples as $N_{test} \times T_1$ for $(N_1,T_1) = (120, 2)$ and $N_{test} \times T_2$ for $(N_2,T_2) = (80, 3)$  with the choice of $N_{test} = 50$. Table~ \ref{Tab2:outliers-rmse} displays the RMSEs obtained from the predictive model constructed by the estimated regression coefficients using all the estimators considered in this study. The results provided by the robust estimators are similar and also, better than traditional LS method especially for concentrated leverage points. However, proposed Tukey's bisquare and Exponential loss functions based estimators have slightly better performance among the robust estimators, in general.

\section{Case Study} \label{Sec:7}

In this section, we compare the performances of all the estimators to check the supremacy of proposed M-estimation procedures based on automatic selection of the tuning parameters via a real macroeconomic application. The data set consists of a total of 342 (N = 18, T = 19) annual observations covering a cross-section of 18 OECD countries over the period 1960-1978 (available in the \texttt{R} package \texttt{plm}, please see \cite{Croissant2008}). For this panel, the gasoline demand model studied by \cite{Baltagi1983} is as follows
\begin{equation*}
\ln GC_{it} = \alpha_i + \beta_1 \ln IPC_{it} + \beta_2 \ln GP_{it}  + \beta_3 \ln CS_{it} + \varepsilon_{it}
\end{equation*}
where $GC$ represents the gasoline consumption per car as dependent variable, $IPC$, $GP$ and $CS$, respectively, denote real income per capita, real gasoline price and car stock per capita as independent variables. The index $i = 1, 2, \cdots, 18$ and $t = 1, 2, \cdots, 19$ denote the OECD countries and years, respectively; see \cite{Baltagi1983} for more detailed description of the data set. Figure~\ref{Fig:3} shows the scatter plots of the logarithm of response variable against the logarithms of explanatory variables. From Figure~\ref{Fig:3}, we observe that there exist outliers in the response and the independent variables. Table~\ref{Tab3:gasoline_consumption} reports the estimates of individual coefficients and standard errors of the estimates. We can see from Table~\ref{Tab3:gasoline_consumption} that the robust methods yield more efficient estimates than the traditional LS method. The WMS estimators produce more precise estimates of individual coefficients among the robust procedures. However, the proposed methods, especially Exponential loss function based estimator, provide competitive results in estimating $\beta_2$ and $\beta_3$. 

To investigate the predictive performances of the LS, WMS and proposed robust methods further, the RMSE of predicted and observed values of the logarithms of the gasoline consumption are calculated by dividing the dataset into the two parts as mentioned in Section \ref{Sec:6.3}. For this purpose, the randomly selected 15 countries ($N-N_{test} = 15$) are used to build a predictive model, and the logarithms of the gasoline consumption of remaining 3 countries ($N_{test} = 3$) are predicted by using the estimated model coefficients. This process is performed $S = 100$ times, and for each time, RMSEs are obtained by using the Equation~(\ref{Eq:rmse}). The results are presented in Figure~\ref{Fig:4}. This figure demonstrates that all the robust methods have better predictive ability compared to the traditional LS method. Furthermore, our proposed robust methods yield considerably better predictions than WMS method.

\section{Conclusions} \label{Sec:8}
In this paper, we propose robust and efficient estimators based on minimizing loss functions in estimating the parameters of fixed effects panel data models. The proposed procedures uses a data-driven approach for selection of regularization parameters, which is important to determine the necessary level of robustness without sacrificing estimation efficiency in the presence outlying observations and/or blocks. The asymptotic properties of the proposed estimators are investigated in the context of fixed effects linear panel data models.
The finite sample performances of the proposed methods are illustrated via extensive simulation studies and an empirical application, and we compare the results of our methods with those for existing WMS estimator of \cite{Bramati2007} and traditional LS method. Our thorough study demonstrates that the proposed estimators produce more accurate and efficient parameter estimates with better predictions compared to WMS and traditional LS methods under data contamination, and  consistent results with the LS method when no outliers are present in the data and sample size increases.

\section*{References}

\bibliography{mybibfile}

\clearpage

\begin{table}
\centering
\caption{The MSEs for $(N_1,T_1) = (120, 2)$, $(N_2,T_2) = (80, 3)$ in the presence of 5\% and 10\% random and concentrated contamination by setting the number of outliers as $m= 12, 24$}
\scriptsize
\begin{tabular}{l c c c c c c c c}
\hline
($N_1$, $T_1$) = (120, 2) & \multicolumn{4}{c}{Random contamination} & \multicolumn{4}{c}{Concentrated contamination} \\
\cmidrule{1-1}
Method  & \multicolumn{2}{c}{Vertical outliers} & \multicolumn{2}{c}{Leverage points} & \multicolumn{2}{c}{Vertical outliers} & \multicolumn{2}{c}{Leverage points} \\ 
& 5\% & 10\% & 5\% & 10\% & 5\% & 10\% & 5\% & 10\% \\
\hline
LS & 1.651 & 2.579 & 7.688 & 9.116 & 2.985 & 5.916 & 25.231 & 30.581 \\
Huber & 0.710 & 0.941 & 0.678 & 0.942 & 0.746 & 0.987 & 0.771 & 0.963 \\
Tukey & 0.609 & 0.842 & 0.614 & 0.828 & 0.577 & 0.642 & 0.569 & 0.630 \\
Exponential & 0.637 & 0.848 & 0.618 & 0.823 & 0.595 & 0.668 & 0.602 & 0.673 \\
WMS & 0.762 & 1.571 & 0.727 & 1.547 & 1.009 & 3.951 & 1.102 & 3.644 \\
\hline 
($N_2$, $T_2$) = (80, 3) \\
\cmidrule{1-1}
LS & 1.172 & 1.982 & 7.464 & 8.830 & 2.467 & 4.150 & 25.422 & 30.490 \\
Huber & 0.575 & 0.806 & 0.606 & 0.824 & 0.599 & 0.902 & 0.652 & 0.922 \\
Tukey & 0.522 & 0.713 & 0.559 & 0.719 & 0.467 & 0.649 & 0.530 &  0.642 \\
Exponential & 0.517 & 0.706 & 0.558 & 0.709 & 0.477 & 0.695 & 0.542 & 0.704 \\
WMS & 0.510 & 0.755 & 0.522 & 0.773 & 0.460 & 0.926 & 0.485 & 0.947 \\
\hline
\end{tabular}
\label{Tab1:outliers-N-120-80-T-2-3-mse}
\end{table}

\begin{table}
\centering
\caption{The RMSEs of predicted values for $(N_1,T_1) = (120, 2)$, $(N_2,T_2) = (80, 3)$ in the presence of 5\% and 10\% random and concentrated contamination by setting the number of outliers as $m= 12, 24$}
\scriptsize
\begin{tabular}{l c c c c c c c c}
\hline
($N_1$, $T_1$) = (120, 2) & \multicolumn{4}{c}{Random contamination} & \multicolumn{4}{c}{Concentrated contamination} \\
\cmidrule{1-1}
Method  & \multicolumn{2}{c}{Vertical outliers} & \multicolumn{2}{c}{Leverage points} & \multicolumn{2}{c}{Vertical outliers} & \multicolumn{2}{c}{Leverage points} \\ 
& 5\% & 10\% & 5\% & 10\% & 5\% & 10\% & 5\% & 10\% \\
\hline
LS & 4.815 & 4.918 & 5.134 & 5.227 & 4.923 & 5.128 & 6.078 & 6.332 \\
Huber & 4.744 & 4.788 & 4.722 & 4.753 & 4.738 & 4.764 & 4.702 & 4.761 \\
Tukey & 4.739 & 4.782 & 4.718 & 4.744 & 4.722 & 4.734 & 4.687 & 4.730 \\
Exponential & 4.740 & 4.784 & 4.718 & 4.742 & 4.724 & 4.736 & 4.690 & 4.734 \\
WMS & 4.748 & 4.838 & 4.727 & 4.801 & 4.766 & 4.992 & 4.732 & 4.968 \\
\hline 
($N_2$, $T_2$) = (80, 3) \\
\cmidrule{1-1}
LS & 5.560 & 5.645 & 6.014 & 6.081 & 5.650 & 5.815 & 7.008 & 7.325 \\
Huber & 5.508 & 5.529 & 5.531 & 5.540 & 5.480 & 5.533 & 5.482 & 5.507\\
Tukey & 5.503 & 5.518 & 5.526 & 5.529 & 5.468 & 5.510 & 5.471 & 5.483 \\
Exponential & 5.503 & 5.518 & 5.526 & 5.528 & 5.469 & 5.515 & 5.472 & 5.489 \\
WMS & 5.503 & 5.524 & 5.523 & 5.536 & 5.470 & 5.547 & 5.466 & 5.522 \\
\hline
\end{tabular}
\label{Tab2:outliers-rmse}
\end{table}

\begin{figure}[!htbp]
  \centering
  \includegraphics[width=8cm]{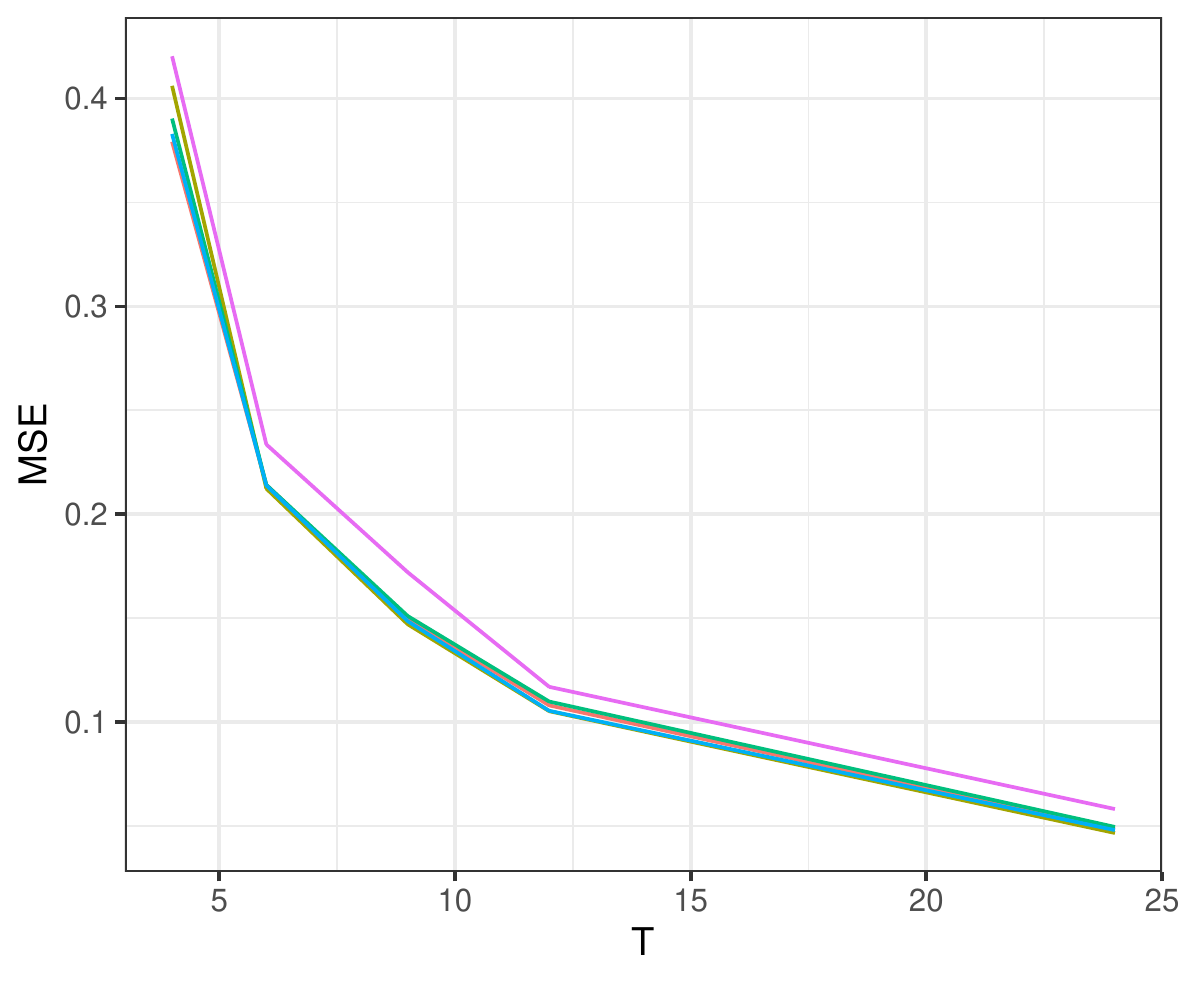}
  \includegraphics[width=8cm]{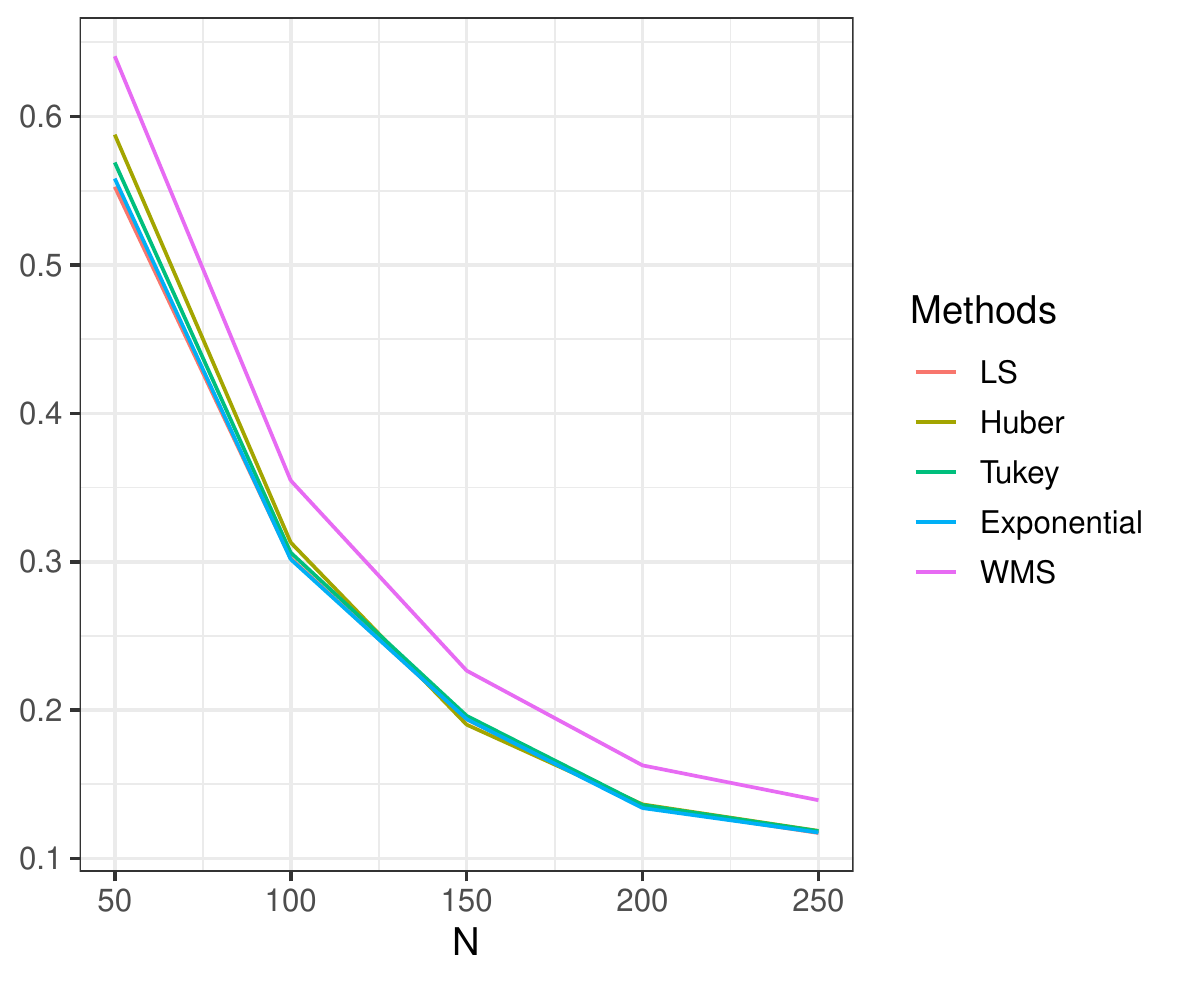}
  \caption{The MSEs of all estimators when $T = 4, 6, 9, 12, 24$ for fixed cross-sectional dimension $N = 50$ and $N = 50, 100, 150, 200, 250$ for fixed time dimension $T = 3$ and $\varepsilon_{it} \sim$ $\text{N}(\mu = 0, \sigma^2 = 1)$}
  \label{Fig:1}
\end{figure}

\begin{figure}[!htbp]
  \centering
  \includegraphics[width=5.2cm]{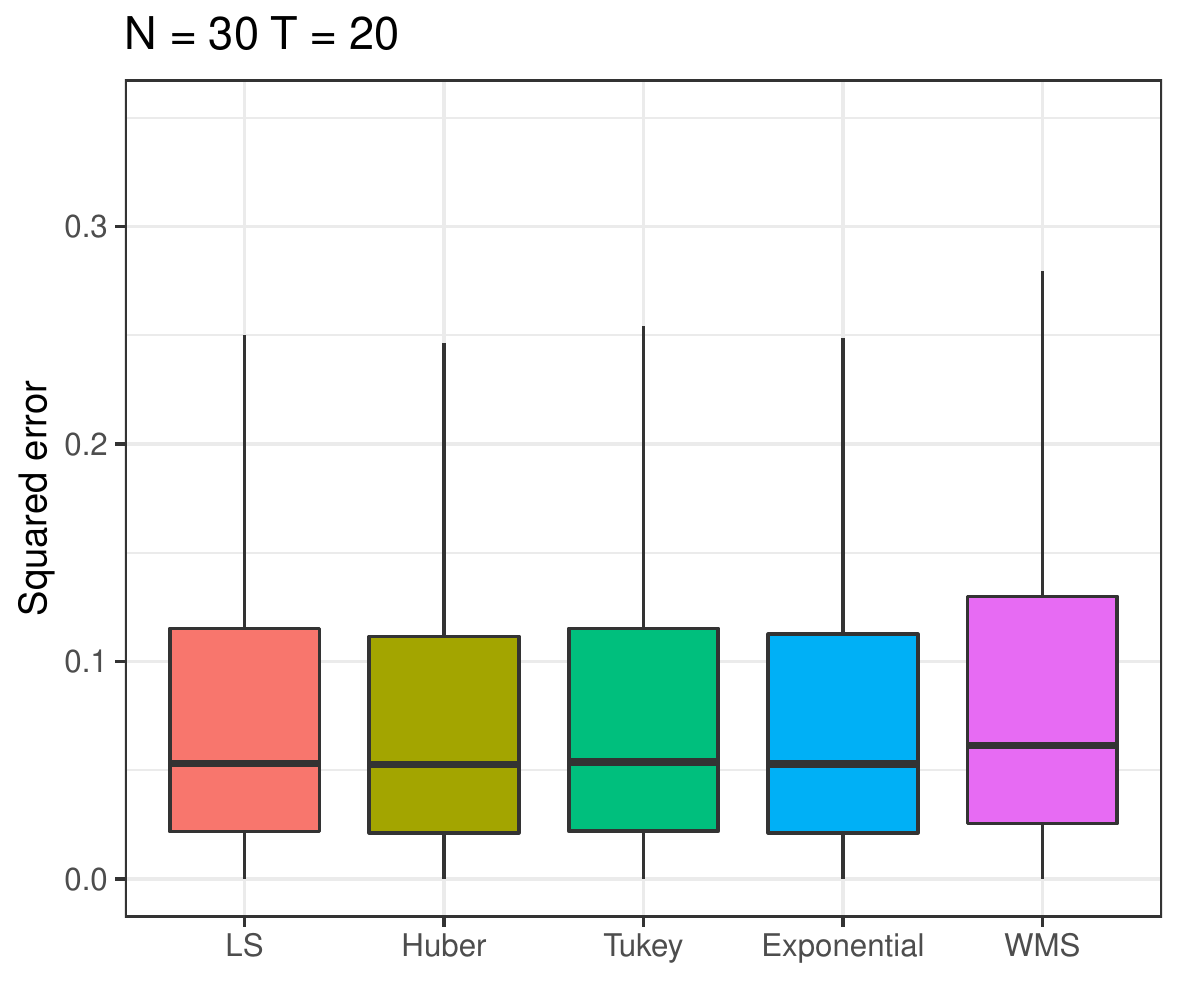}
  \includegraphics[width=5.2cm]{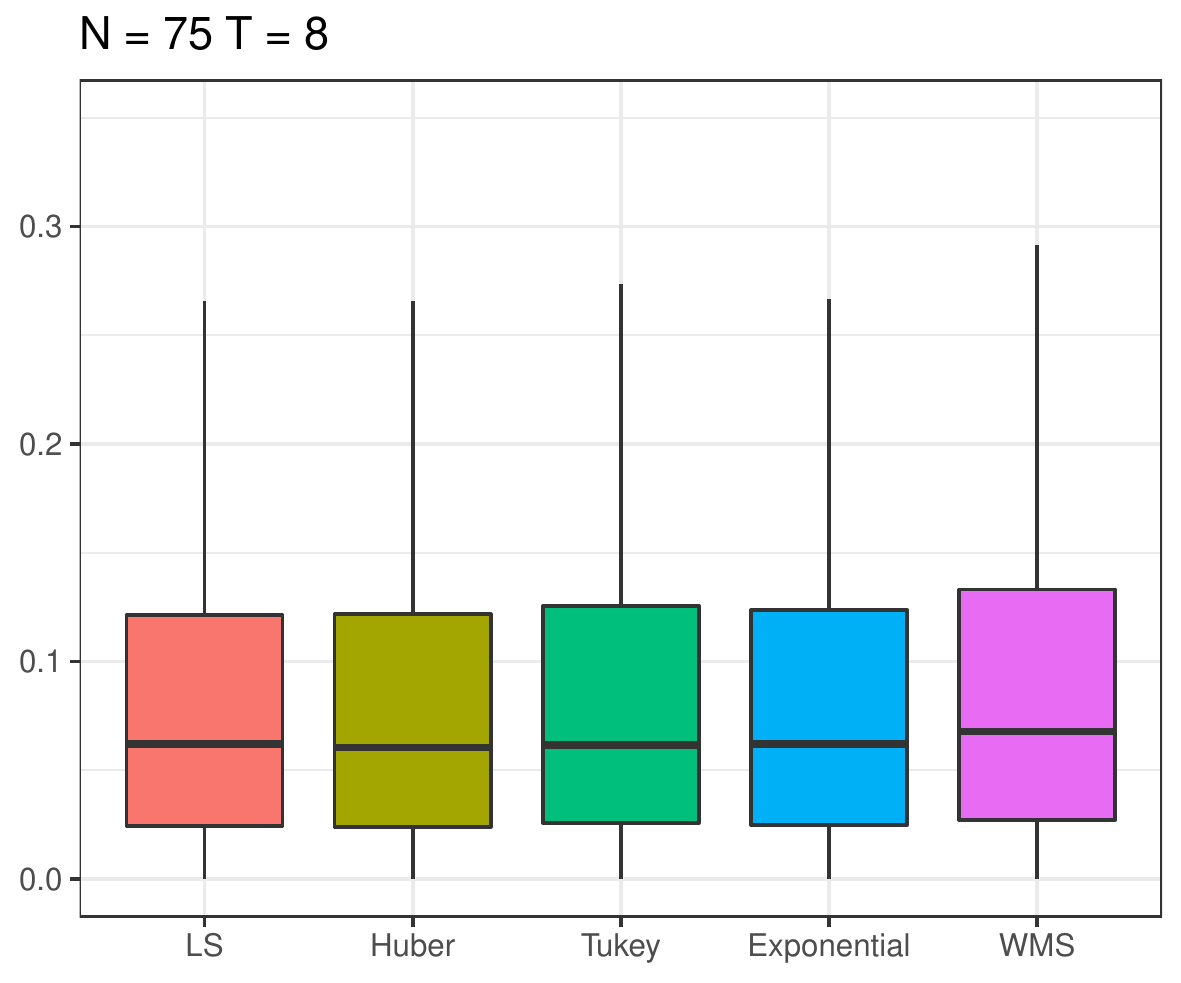}
  \includegraphics[width=5.2cm]{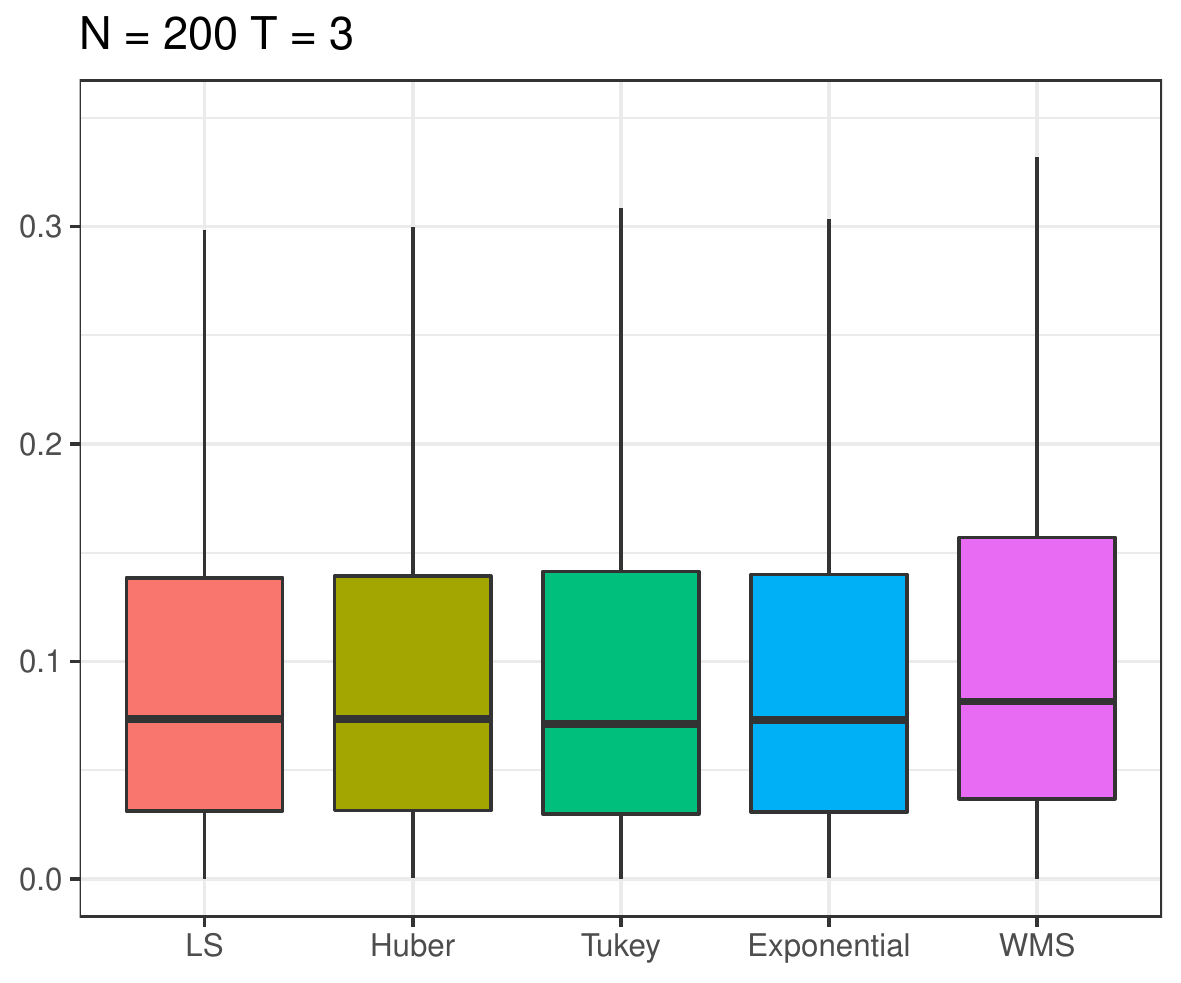}
  \\
  \includegraphics[width=5.2cm]{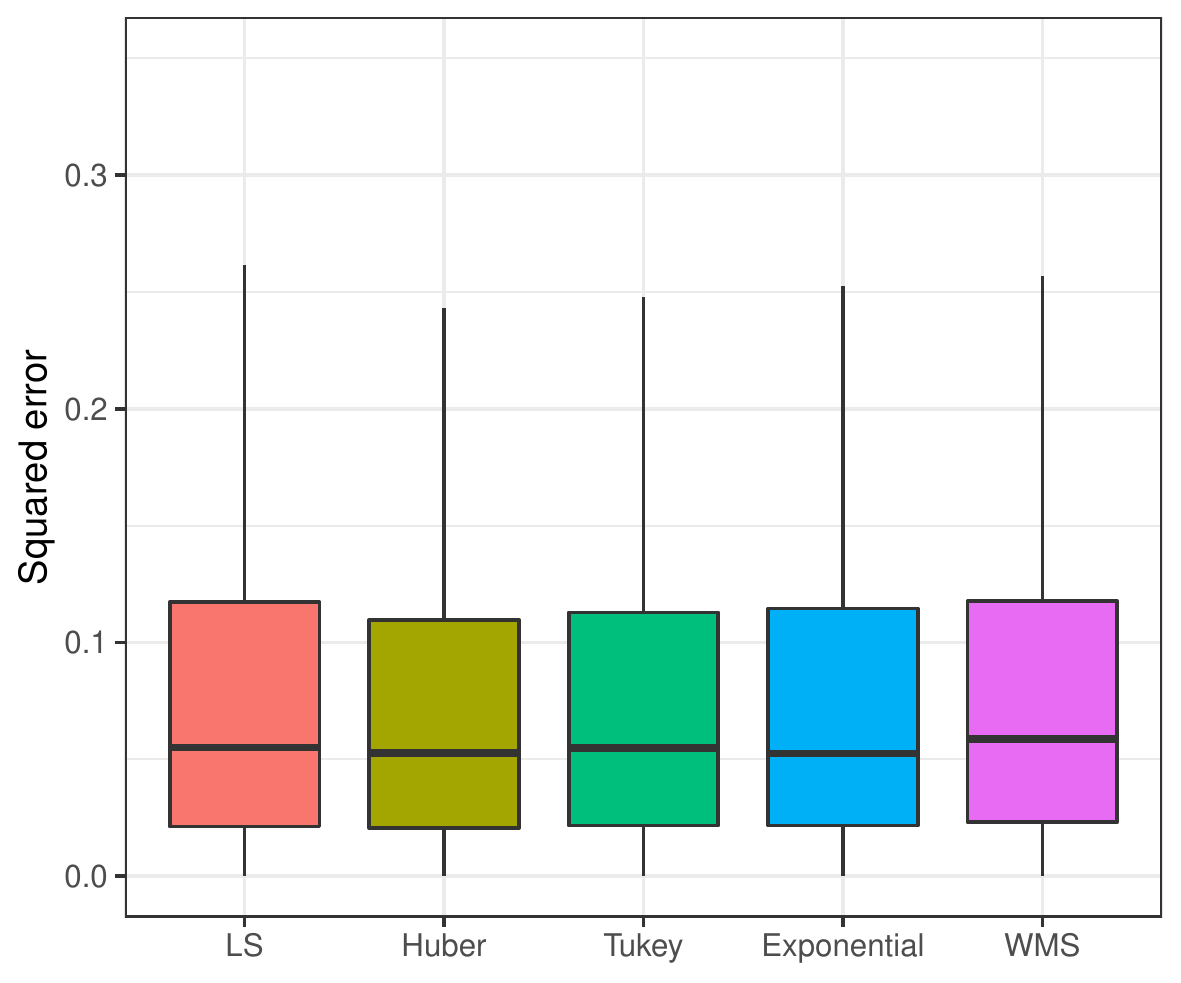}
  \includegraphics[width=5.2cm]{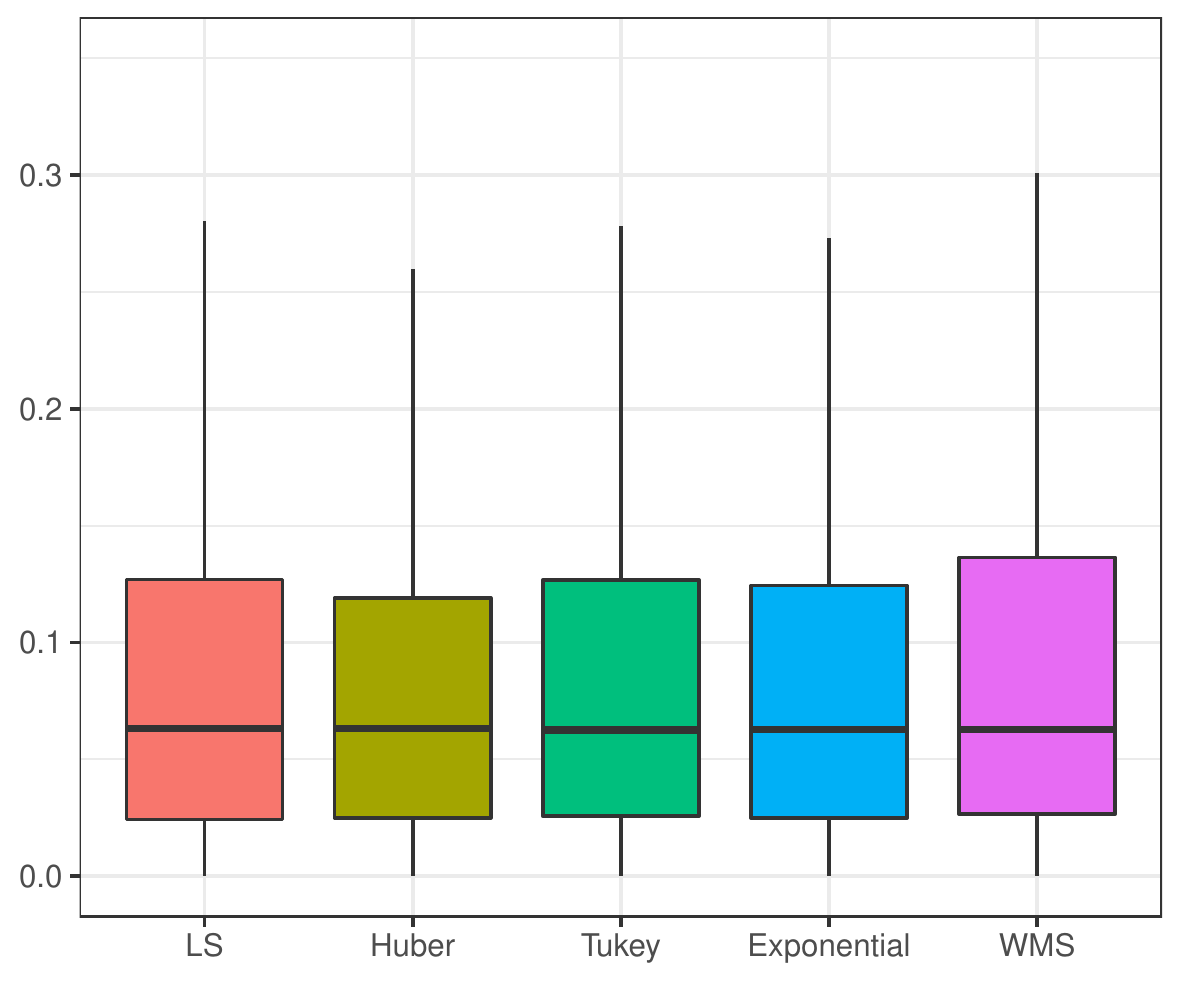}
  \includegraphics[width=5.2cm]{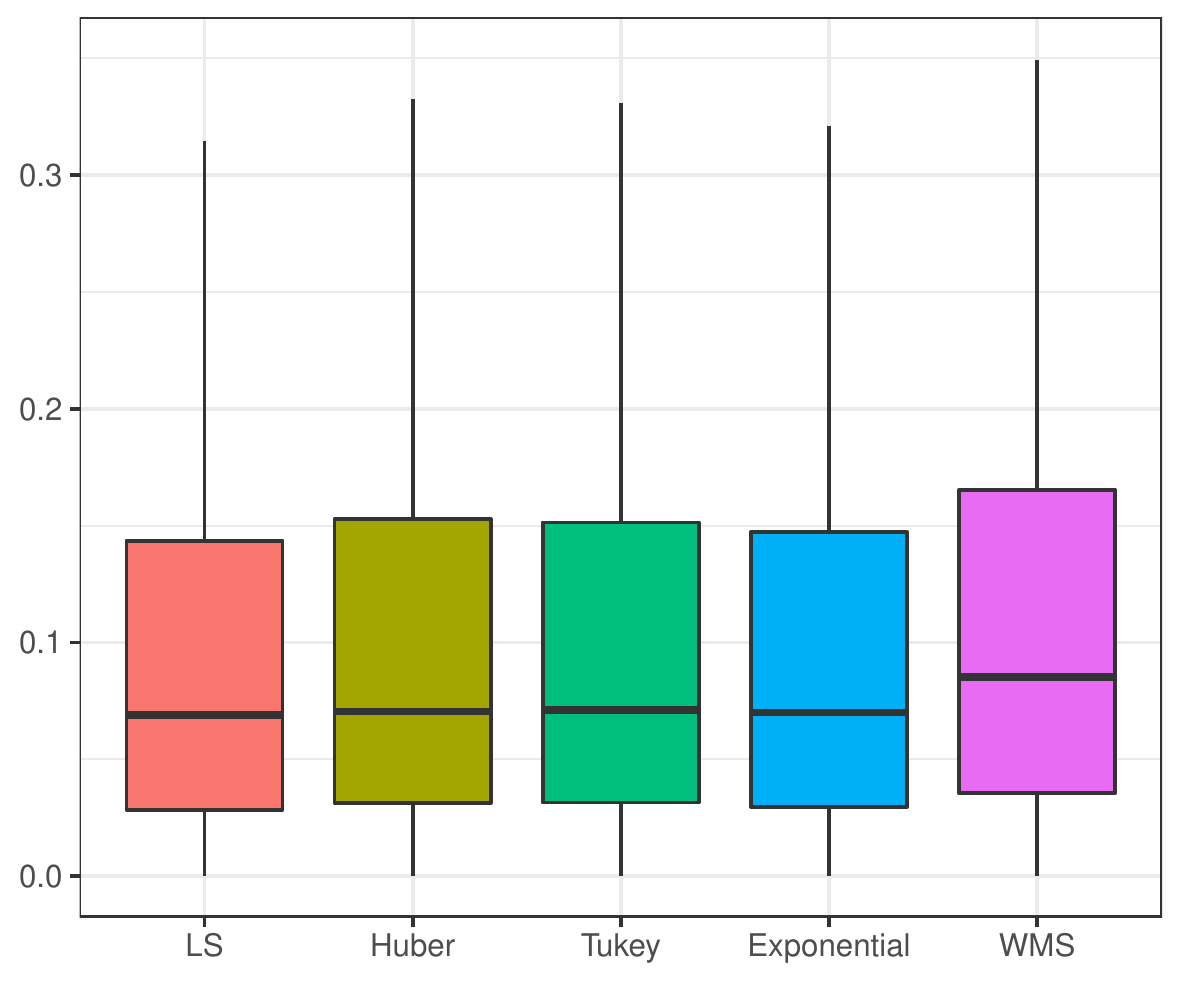}
  \\
  \includegraphics[width=5.2cm]{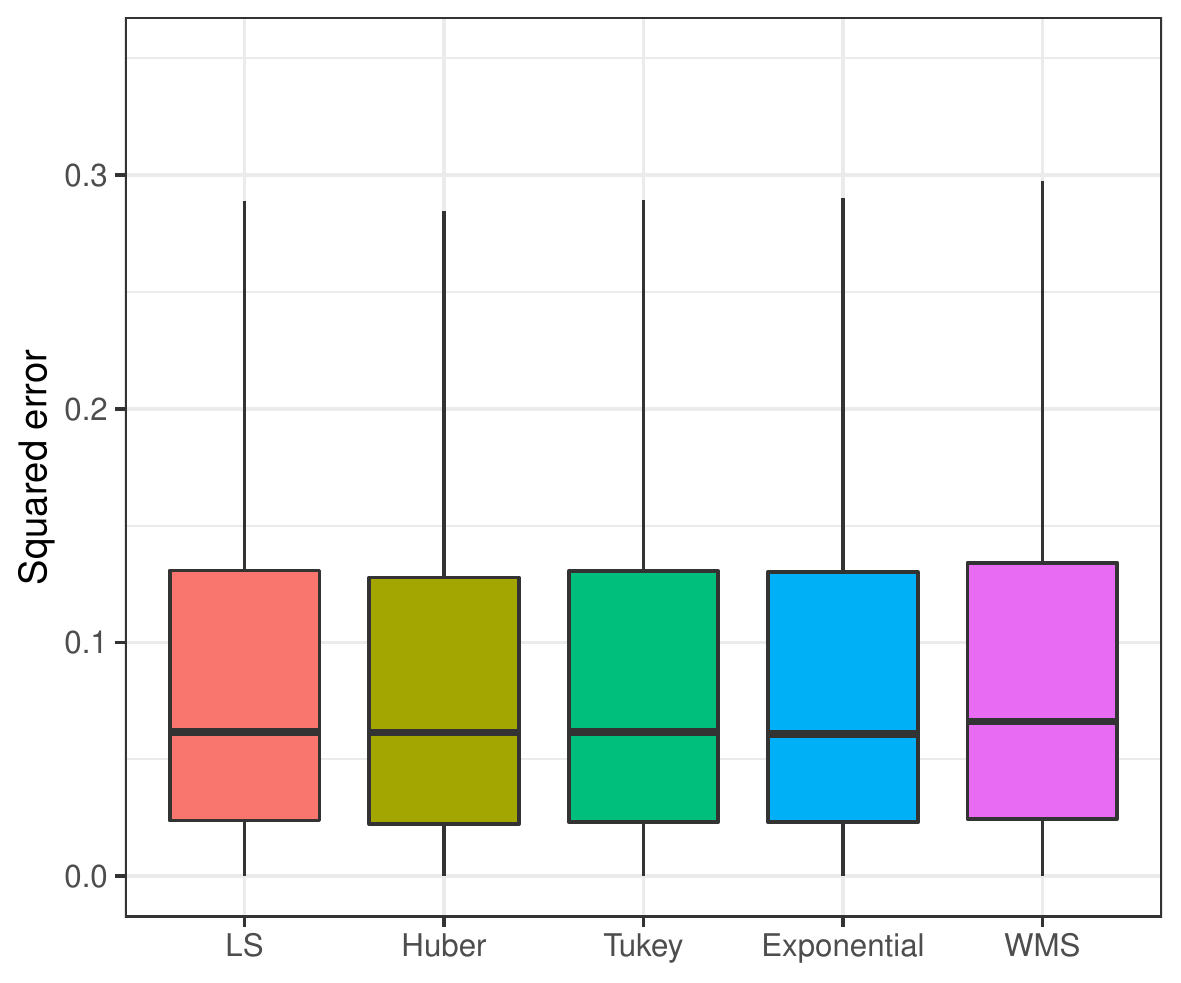}
  \includegraphics[width=5.2cm]{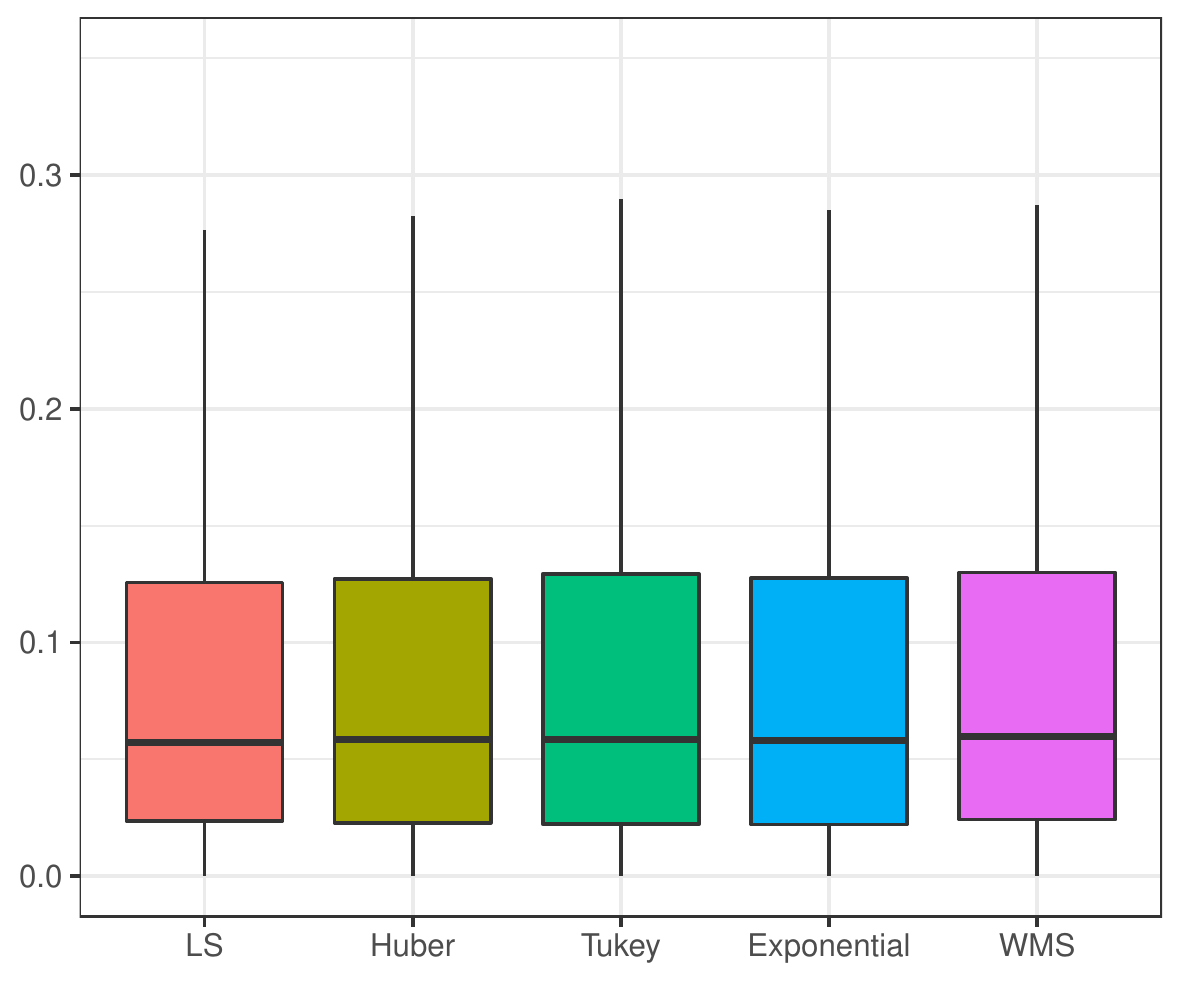}
  \includegraphics[width=5.2cm]{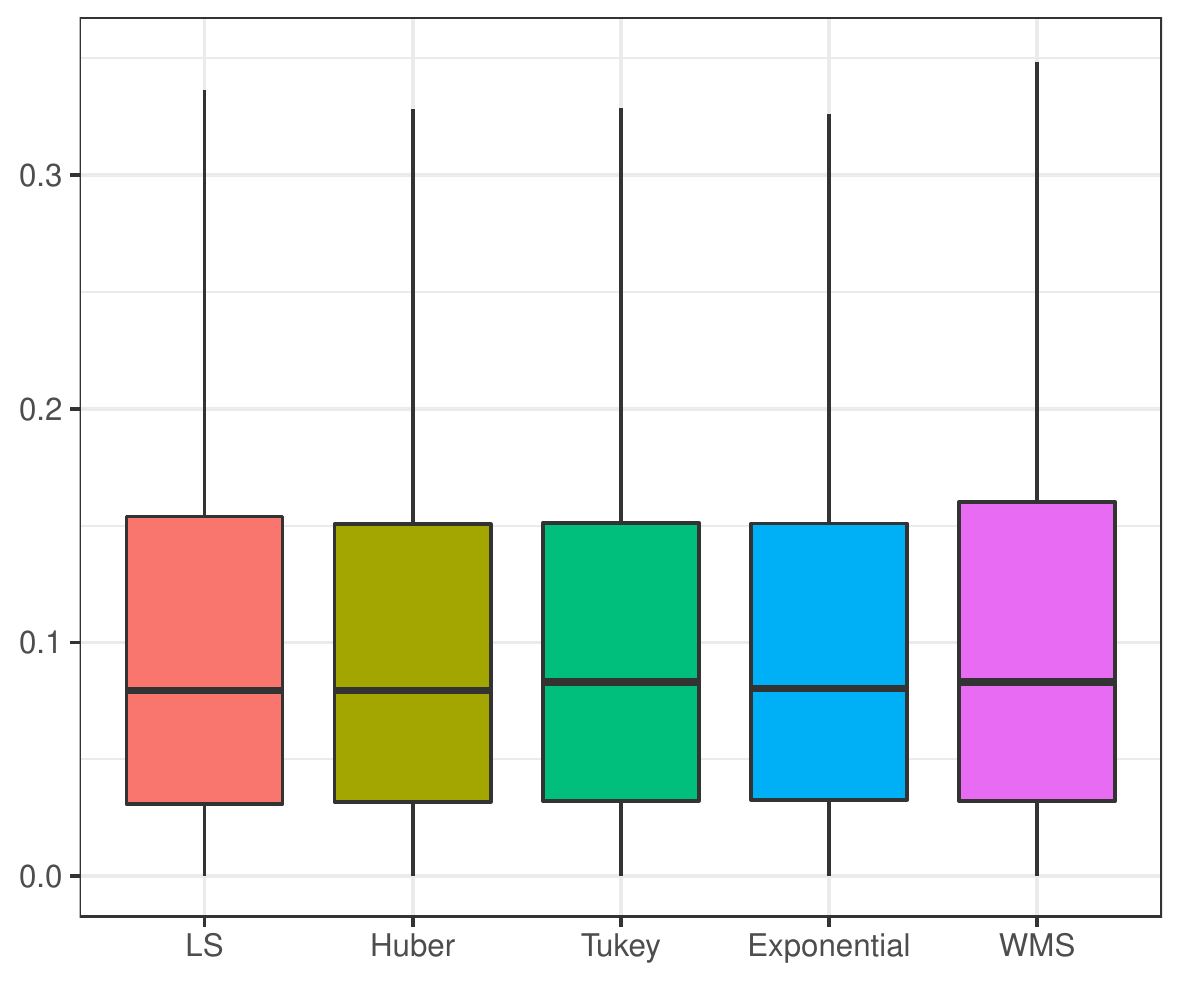}
  \\
  \includegraphics[width=5.2cm]{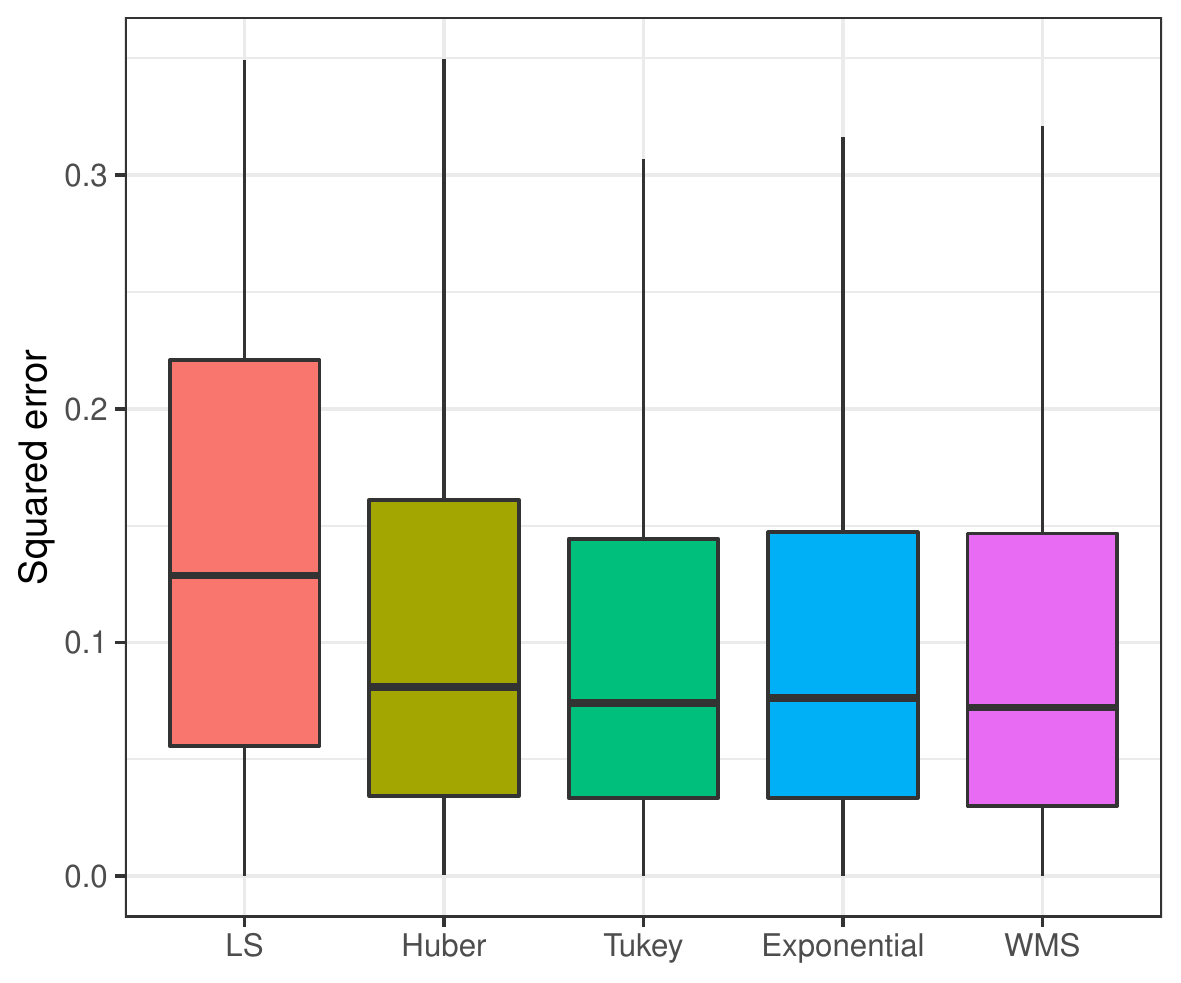}
  \includegraphics[width=5.2cm]{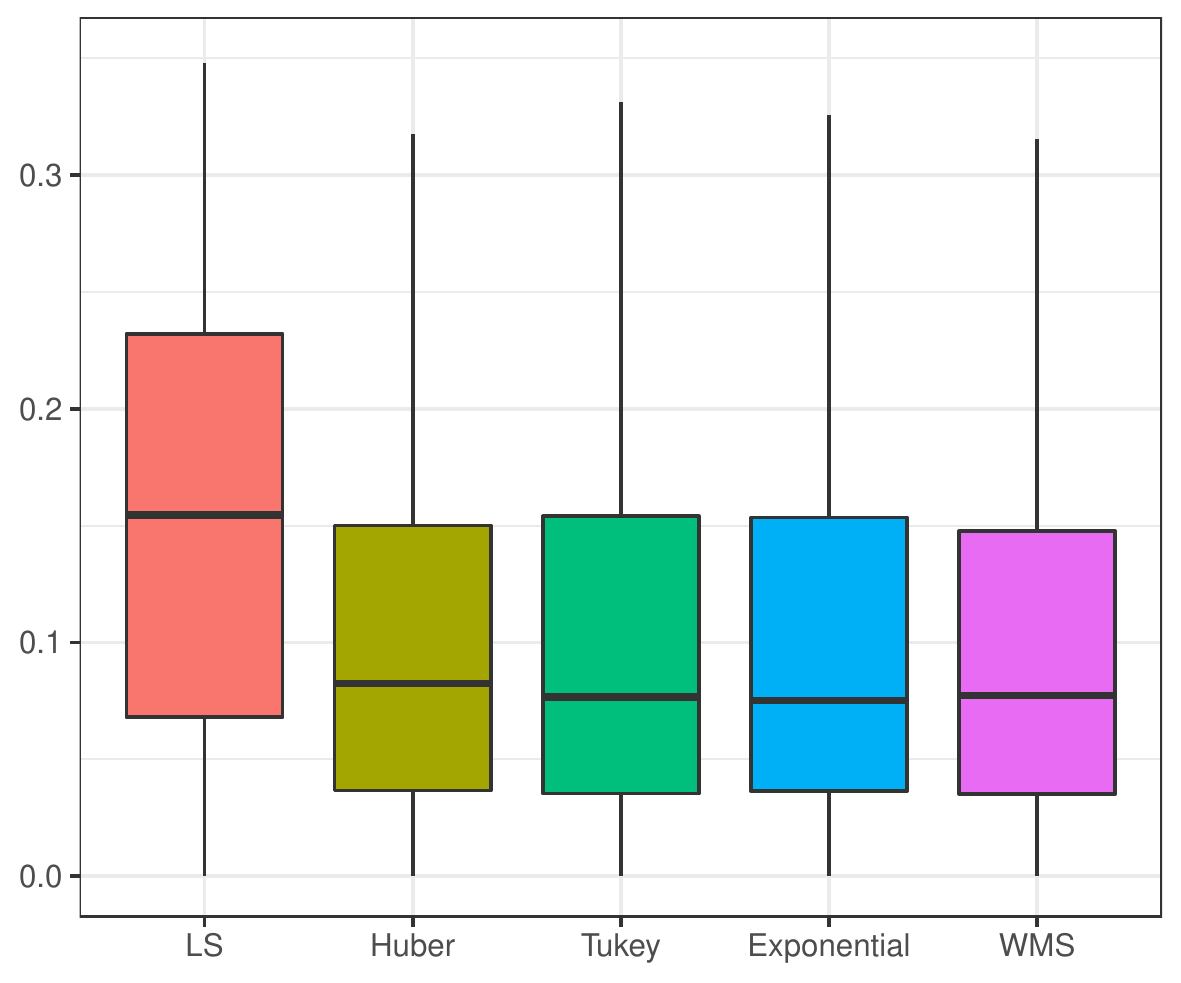}
  \includegraphics[width=5.2cm]{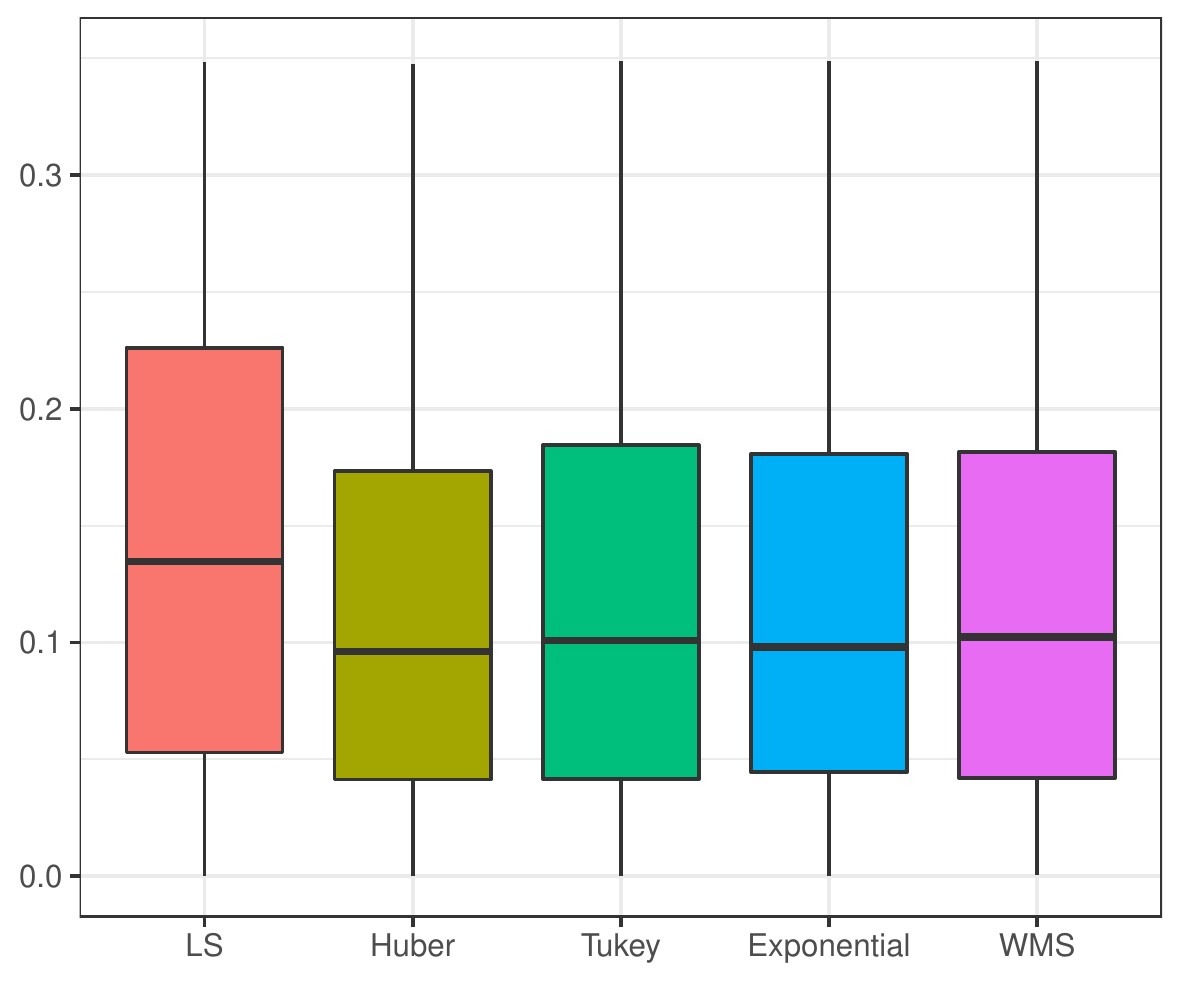}
  \caption{The boxplots of the SEs of all estimators under different error distributions when $(N_1, T_1)=(30,20)$, $(N_2, T_2)=(75, 8)$ and $(N_3, T_3)=(200, 3)$. First row: $N(0,1)$ distribution, second row: $t_5$ distribution, third row: $\chi^2_{(4)}$ distribution and fourth row: $Cauchy(0, 1)$ distribution.}
  \label{Fig:2}
\end{figure}

\begin{figure}[!htbp]
  \centering
  \includegraphics[width=5.2cm]{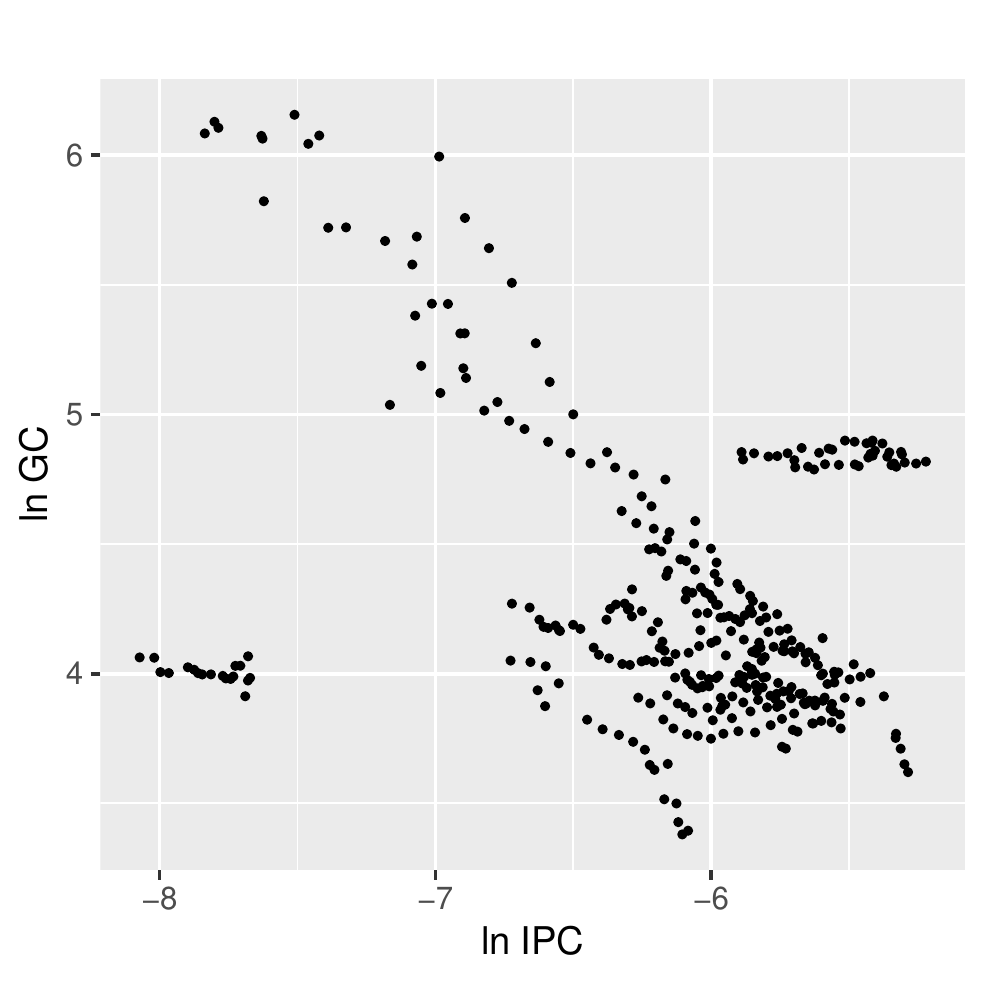}
  \includegraphics[width=5.2cm]{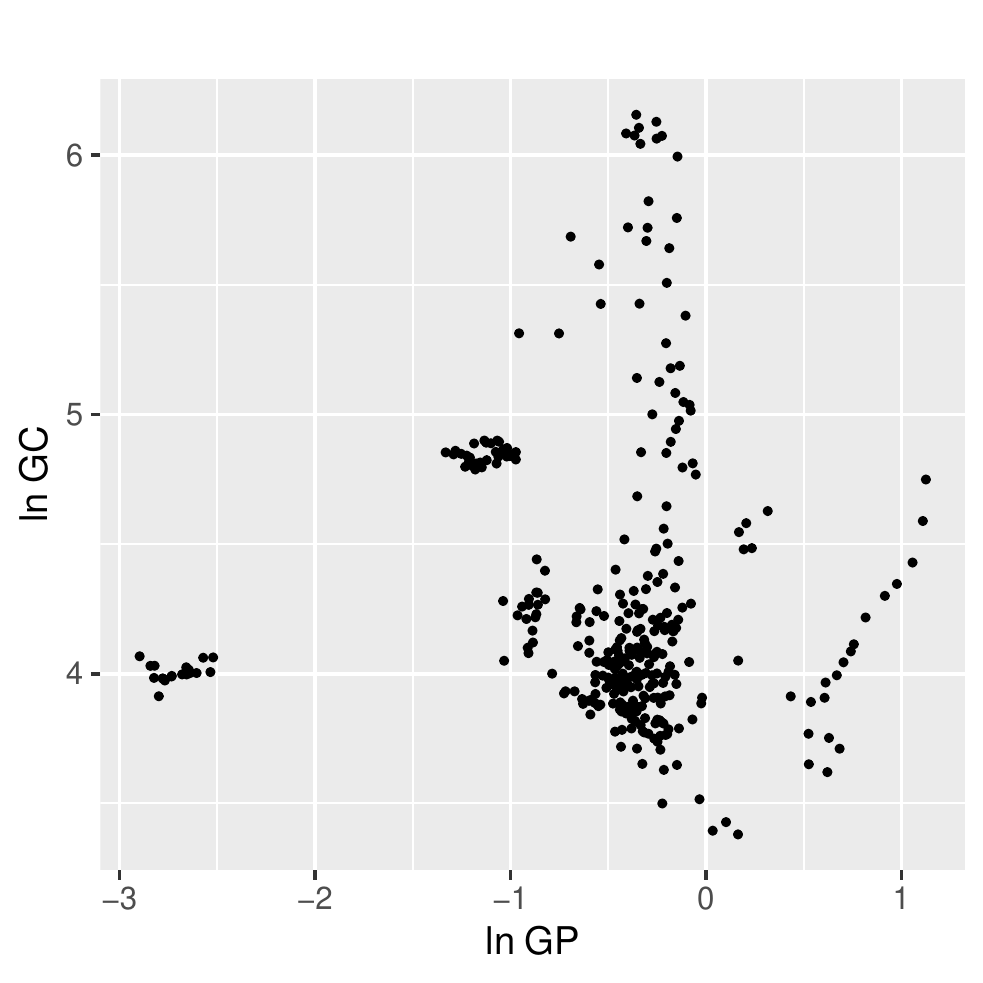} 
  \includegraphics[width=5.2cm]{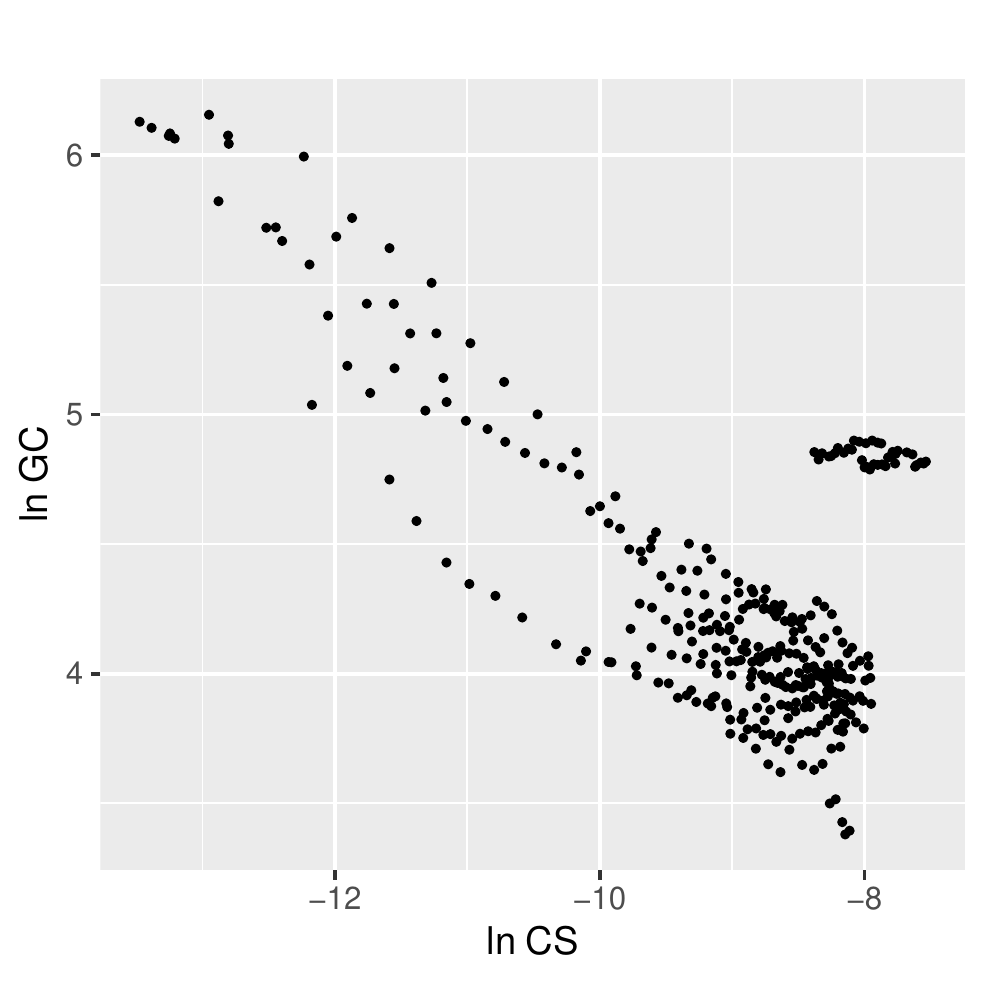} 
  \caption{Scatter plots of the logarithm of gasoline consumption ($\ln GC$) against the logarithms of real income ($\ln IPC$), real gasoline price ($\ln GP$) and car stock ($\ln CS$) for 18 OECD countries.}
  \label{Fig:3}
\end{figure}

\begin{figure}[!htbp]
  \centering
  \includegraphics[width=10cm]{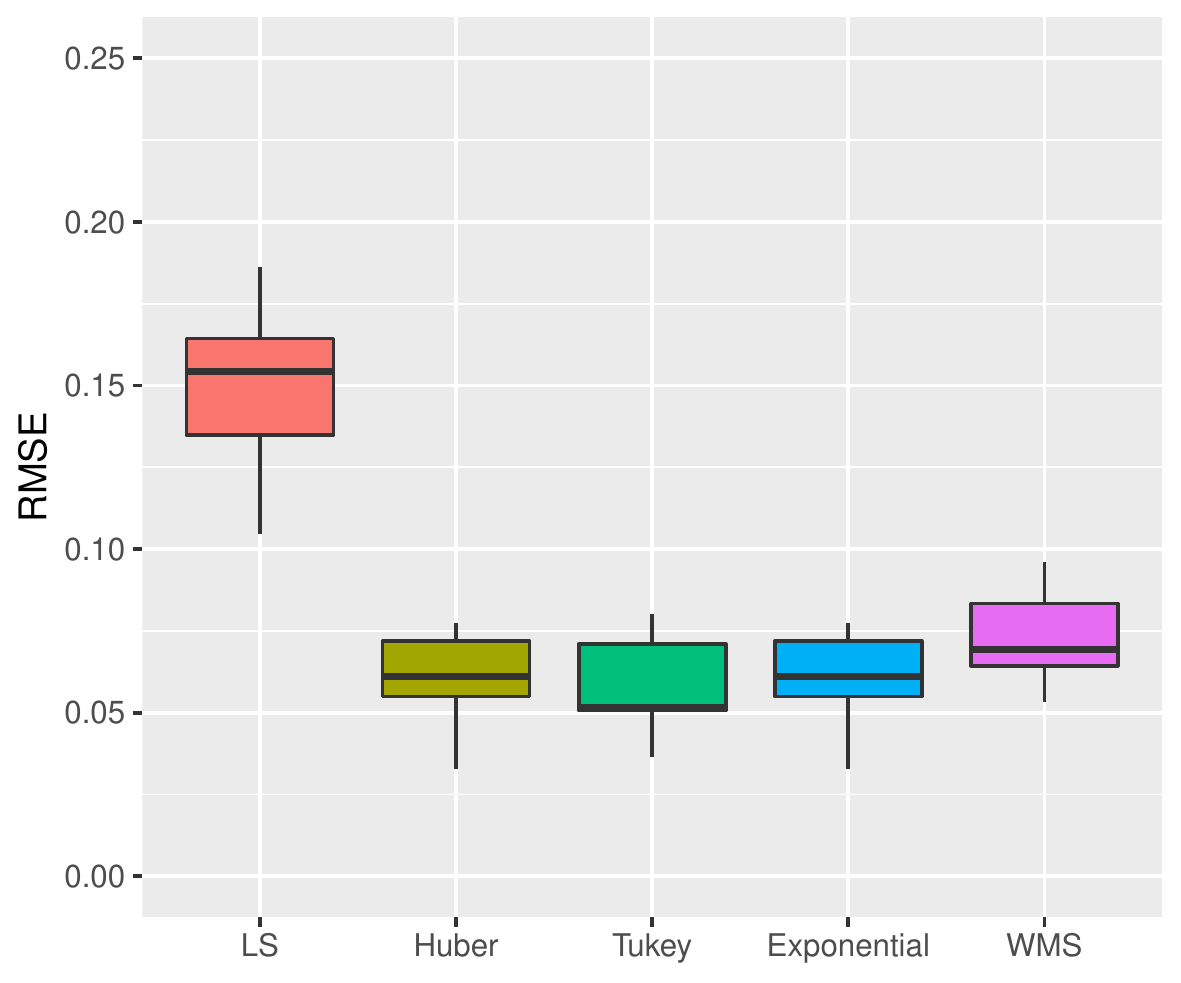}
  \caption{The boxplots of the RMSEs of predicted and observed values of the logarithm of the gasoline consumption for 3 OECD countries.}
  \label{Fig:4}
\end{figure}

\begin{table}
\centering
\caption{Estimates of individual coefficients (upper rows) and standard errors of the estimates (lower rows) for gasoline consumption data of 18 OECD countries over the period 1960-1978}
\scriptsize
\begin{tabular}{l c c c}
\cmidrule{1-4} 
Method & $\beta_1$ & $\beta_2$ & $\beta_3$ \\ 
\cmidrule{1-4} 
LS & -1.043 & -0.264 & 0.113 \\
 & (0.054) & (0.066) & (0.021) \\
Huber & 0.579 &  -0.317 & -0.559 \\
 & (0.048) & (0.029) & (0.019) \\
Tukey & 0.479 & -0.302 & -0.451 \\
 & (0.040) & (0.024) & (0.016) \\
Exponential & 0.482 & -0.307 & -0.452 \\
 & (0.027) & (0.018) & (0.012) \\
WMS & 0.902 & -0.356 & -0.651 \\
 & (0.017) & (0.014) & (0.011) \\
\cmidrule{1-4}  
\end{tabular}
\label{Tab3:gasoline_consumption} 
\end{table}

\end{document}